\documentclass[aps,prl,twocolumn,superscriptaddress]{revtex4-1}

\usepackage[english]{babel}
\usepackage[utf8]{inputenc}
\usepackage{amsmath}
\usepackage{amssymb}
\usepackage{amsthm}
\usepackage{mathtools}
\usepackage{color}
\usepackage{dsfont}
\usepackage{float}
\usepackage[caption = false]{subfig}
\usepackage{graphicx}
\usepackage{soul}
\usepackage{hyperref}

\usepackage[draft,inline,nomargin]{fixme} \fxsetup{theme=color}

\newcommand{\id}{\openone}
\newcommand{\tr}{{\rm tr}}
\newcommand{\ket}[1]{\left|{#1}\right\rangle}
\newcommand{\bra}[1]{\left\langle{#1}\right|}
\newcommand{\braket}[2]{\langle{#1}|{#2}\rangle}
\newcommand{\ketbrad}[1]{\left|{#1}\rangle\!\langle{#1}\right|}
\newcommand{\ketbra}[2]{\left|{#1}\rangle\!\langle{#2}\right|}
\newcommand{\mean}[1]{\langle{#1}\rangle}

\newcommand{\ex}[1]{{\rm e}^{#1}}
\renewcommand{\P}{p}
\newcommand{\Q}{q}
\renewcommand{\vec}[1]{\mathbf{#1}}

\newcommand{\be}{\begin{equation}} 							
\newcommand{\ee}{\end{equation}}
\newcommand{\ba}{\begin{align}}
\newcommand{\ea}{\end{align}}
\newcommand{\bematrix}{\left(\begin{matrix}}
\newcommand{\ematrix}{\end{matrix}\right)}

\DeclareMathOperator{\Tr}{Tr}

\newtheorem{theorem}{Theorem}
\newtheorem{lemma}{Lemma}

\begin{document}

\title{Quantum Sequential Hypothesis Testing}
\date{\today}
\author{Esteban Mart\'inez Vargas}
\email{Esteban.Martinez@uab.cat}
\affiliation{F\'isica Te\`orica: Informaci\'o i Fen\`omens Qu\`antics, Departament de F\'isica, Universitat Aut\`onoma de Barcelona, 08193 Bellatera (Barcelona) Spain}
\author{Christoph Hirche}
\email{christoph.hirche@gmail.com}
\affiliation{QMATH, Department of Mathematical Sciences, University of Copenhagen, Universitetsparken 5, 2100 Copenhagen, Denmark}
\author{Gael Sent\'is}
\email{Gael.Sentis@uab.cat}
\affiliation{F\'isica Te\`orica: Informaci\'o i Fen\`omens Qu\`antics, Departament de F\'isica, Universitat Aut\`onoma de Barcelona, 08193 Bellatera (Barcelona) Spain}
\author{Michalis Skotiniotis}
\email{michail.skoteiniotis@uab.cat}
\affiliation{F\'isica Te\`orica: Informaci\'o i Fen\`omens Qu\`antics, Departament de F\'isica, Universitat Aut\`onoma de Barcelona, 08193 Bellatera (Barcelona) Spain}
\author{Marta Carrizo}
\affiliation{F\'isica Te\`orica: Informaci\'o i Fen\`omens Qu\`antics, Departament de F\'isica, Universitat Aut\`onoma de Barcelona, 08193 Bellatera (Barcelona) Spain}
\author{Ramon Mu\~noz-Tapia}
\email{Ramon.Munoz@uab.cat}
\affiliation{F\'isica Te\`orica: Informaci\'o i Fen\`omens Qu\`antics, Departament de F\'isica, Universitat Aut\`onoma de Barcelona, 08193 Bellatera (Barcelona) Spain}
\author{John Calsamiglia}
\email{John.Calsamiglia@uab.cat}
\affiliation{F\'isica Te\`orica: Informaci\'o i Fen\`omens Qu\`antics, Departament de F\'isica, Universitat Aut\`onoma de Barcelona, 08193 Bellatera (Barcelona) Spain}

\begin{abstract}
We introduce  sequential analysis in quantum information processing, 
by focusing  on the fundamental task of quantum hypothesis testing. 
In particular our goal is to discriminate between two arbitrary 
quantum states with a prescribed error threshold, $\epsilon$, when 
copies of the states can be required on demand.  We obtain 
ultimate lower bounds on the average number of copies needed to 
accomplish the task. We give a 
block-sampling strategy that allows to achieve the lower bound for 
some classes of states. The bound is optimal in both the symmetric as 
well as the asymmetric setting in the sense that it requires the 
least mean number of copies out of all other procedures, including 
the ones that fix the number of copies ahead of time. For qubit states we derive explicit 
expressions for the minimum average number of copies and show that 
a sequential strategy based on fixed local measurements outperforms 
the best collective measurement on a predetermined number of copies. 
Whereas for general states the  number of copies increases as 
$\log 1/\epsilon$,  for pure states sequential strategies 
require a finite average number of samples even in the case of 
perfect discrimination, i.e., $\epsilon=0$.
\end{abstract}

\maketitle

\paragraph{Introduction.} 
Statistical inference permeates almost every human endeavor, from science and engineering all the way through to economics, 
finance, and medicine. The perennial dictum in such inference tasks has been to optimize performance---often quantified by 
suitable cost functions---given a fixed number, $N$, of relevant resources~\cite{lehmann_2006, Helstrom1976}. 
This approach often entails the practical drawback that all $N$ resources need to be batch-processed before a good inference 
can be made. Fixing the number of resources ahead of time does not reflect the situation that one encounters in many real-life 
applications that might require an online, early, inference--such as change-point 
detection~\cite{Tartakovsky2014,Sentis2016,Sentis2017,Sentis2018}, or where additional data may be obtained on demand if 
the required performance thresholds are not met.

Sequential analysis~\cite{wald1973sequential} is a statistical 
inference framework designed to address these shortcomings. 
Resources are processed on-the-fly, and with each new measured 
unit a decision to stop the experiment is made depending on whether 
prescribed tolerable error rates (or other cost functions) are met; 
the processing is continued otherwise.  
Since the decision to stop is solely based on previous measurement 
outcomes, the size $N$ of the experiment is not predetermined but 
is, instead, a random variable. 
A sequential protocol is  deemed optimal if it requires the least 
\emph{average number of resources} among all statistical tests 
that guarantee the same performance thresholds. 
For many classical statistical inference tasks it is known that 
sequential methods can attain the required thresholds with 
substantially lower average number of samples  than any statistical 
test based on a predetermined number of samples \cite{wald1973sequential}. The ensuing 
savings in resources, and the  ability to take actions in real-time,  
have found applications in a wide range of fields~\cite{leung_lai_sequential_2001,Tartakovsky2014}. Extending  sequential analysis  to the quantum setting is of fundamental interest, and
with near-term quantum technologies on the verge of impacting the global market, the versatility and resource efficiency that sequential protocols  provide for quantum information processing is highly desirable.

In this paper we consider the discrimination of two arbitrary finite dimensional
quantum states~\cite{bae2015quantum}, $\rho$ (corresponding to the null 
hypothesis $H_{0}$) and $\sigma$ (corresponding to the alternative 
hypothesis $H_{1}$),  in a setting where a large number of copies 
can be used in order to meet a desired error threshold $\epsilon$. 
A first step in this direction was taken in Ref.~\cite{Slussarenko2017MinASN}, which considers the particular case where $\rho$ and $\sigma$ are pure states and restricts the analysis to specific local measurement strategies. 
Here, we address the problem in full generality, including arbitrary states, weak and collective measurements.
For collective strategies involving a large \emph{fixed} number of copies
the relation between this number and the error $\epsilon$ is 
$N\sim-\frac{1}{\xi}\log \epsilon$~\footnote{By 
    $f(\epsilon)\sim g(\epsilon)$ we mean asymptotic equivalence 
    $\lim_{\epsilon\to 0}f(\epsilon)/g(\epsilon)=1$}, where the rate $\xi$ 
depends on the pair of hypotheses and on the precise setting as we 
explain shortly. We show that one can significantly reduce the expected number of copies, $\mean{N}$, by considering sequential strategies where copies are provided on demand.  We give the ultimate lower-bounds as a single-letter 
expression of the form,
 \be
\label{eq:introlwb}
\mean{N}_{0}\geq-\frac{\log\epsilon}{D( \rho\|  \sigma)}+O(1)\;\mathrm{,}
\mean{N}_{1}\geq-\frac{\log\epsilon}{D( \sigma\|  \rho)}+O(1)
\ee
for $\epsilon\ll 1$, where $\mean{N}_{\nu}$ is the mean number of copies given the true hypothesis is 
$\nu\in\{0,1\}$ and $D(\rho\|\sigma)=\tr \rho(\log\rho-\log\sigma)$ is the quantum relative entropy.
In addition, we provide upper bounds which, for the worst case 
 $N_{\mathrm{wc}}=\max\{\mean{N}_{0},\mean{N}_{1}\}$, are achievable
 for some families of states. 


Specifically, we consider quantum hypothesis testing in a scenario where one can guarantee that for \emph{each} 
realization of the test  the conditional probability of 
correctly identifying each of the hypotheses is above a given 
threshold. This scenario, first introduced in~\cite{Slussarenko2017MinASN}, can be considered genuinely 
sequential since such \emph{strong error}  conditions cannot be 
generally met in a deterministic setting.  The proof method can be 
easily extended to the more common asymptotic  \emph{symmetric} 
and \emph{asymmetric} scenarios involving the usual \emph{type I} 
(or \emph{false positive}) and \emph{type II} (or \emph{false negative}) errors. 
We give the optimal scaling of the mean number of copies when 
the thresholds for
either one or both types of errors are asymptotically small.

Before proceeding, let us briefly review these 
fundamental hypothesis testing scenarios, which come about from 
the relative importance one places between type I error---the error of 
guessing the state to be $\sigma$ when the true state is $\rho$ 
whose probability we denote by $\alpha=P(\hat H_{1}|\rho)$---and  
type II error---the error of guessing $\rho$ when the state is 
$\sigma$ whose probability is $\beta=P(\hat H_{0}|\sigma)$. 
Often, the two types of errors are put on equal footing (\emph{symmetric} scenario) and one seeks to minimise the mean probability of error 
$\bar\epsilon=\eta_{0} \alpha+\eta_{1} \beta$ with $\eta_0,\eta_1=1-\eta_0$ the prior probabilities for each hypothesis. 
The mean error decays exponentially with the number of copies with an optimal rate given by the Chernoff 
distance~\cite{audenaert_discriminating_2007,nussbaum_chernoff_2009}, $\xi_{\rm Ch}=- \inf _{0 \leq s \leq 1} \log\tr(\rho^{1-s} 
\sigma^{s})$.  

Yet, there are asymmetric
 instances,
 e.g.,  in medical trials, 
where the effect of approving an ineffective 
treatment (type-II) is far worse than discarding a potentially 
good one (type-I).  In such cases it is imperative to minimise the type II error whilst maintaining
a finite probability of successfully identifying the null hypothesis, i.e., $p(\hat{H}_0|\rho)=1-\alpha\geq p_s>0$.    
The corresponding optimal error rate for quantum hypotheses is given by quantum Stein's 
lemma \cite{hiai_proper_1991,Nagaoka1999StrongConv}, 
$\beta\sim\ex{-\xi_{S} N}$ where $\xi_{S}=D(\rho||\sigma)$ is the 
quantum relative entropy. If, on the other hand, we require 
the type I error to decay exponentially, i.e., $\alpha\leq \ex{-r N}$ 
for some rate $r$, then the optimal rate is given by the quantum Hoeffding bound 
\cite{hayashi_error_2007,nagaoka_converse_2006}.  These optimal error rates for strategies with fixed number of copies
have found applications in quantum Shannon theory~\cite{wilde_quantum_2017}, quantum 
illumination~\cite{lloyd_enhanced_2008}, and provide operational meaning to abstract information 
measures~\cite{audenaert_asymptotic_2007,calsamiglia_quantum_2008,berta2017composite}.

What the above results also show is that for $N$ fixed there is a trade-off between 
the probabilities of committing either error.  The advantage of sequential analysis is that it provides strategies capable of 
minimising the average number of copies when \emph{both} errors are bounded, and yields higher asymptotic rates in each of the 
settings described above.

\paragraph{Fixed local measurements.}
We begin by considering the case when each quantum system is 
measured with the same measurement apparatus $E$, giving rise to identically distributed samples of a classical probability distribution. This strategy 
has the advantage of being easily implementable, and that it lets us  
introduce the classical sequential analysis framework. 
Specifically, the optimal classical sequential test, for both the strong 
error as well as the symmetric and asymmetric setting, is known to 
be the \emph{Sequential Probability Ratio Test}  
(SPRT)\cite{Wald1948OptimumChar} which we now review.

After $n$ measurements have been performed, we have a 
string of outcomes $\vec{x}_n=\{x_1,x_2,\ldots,x_n\}$, 
where each element has been sampled effectively from a probability 
distribution determined by the POVM  ${E}=\{E_{x}\}$ and the true 
state of the system, i.e., either ${\P}(x):=\tr(E_x\rho)$, or 
${\Q}(x):=\tr(E_x\sigma)$. For given error thresholds 
$\epsilon_{0}, \epsilon_{1}$, the strong condition demands that 
for each conclusive sequence  the conditional probabilities obey 
either
 \begin{align}
\label{eq:cond0}
P(\rho|\vec{x}_n) &= \frac{\eta_0 \P(\vec{x}_n)}{\eta_0 \P(\vec{x}_n)+\eta_1\Q(\vec{x}_n)} \geq 1-\epsilon_{0},\; {\rm or}  \\
 \label{eq:cond1}
P(\sigma|\vec{x}_n) &= \frac{\eta_1 \Q(\vec{x}_n)}{\eta_0 \P(\vec{x}_n)+\eta_1\Q(\vec{x}_n)} \geq 1-\epsilon_{1}
 \end{align}
where  $\P(\vec{x}_n)=\prod_{k=1}^n \P(x_k)$ since the copies are identical and independent (the same holds for $\Q$). If neither condition is met,  a new copy needs to be requested and we continue measuring. That is, starting at $n=1$ at every step $n$ we check whether
\begin{enumerate}
\item $P(\rho | \vec x_{n}) \geq 1-\epsilon_{0}$, then STOP and accept $H_0$, with guaranteed probability of success  $s_{0}=1-\epsilon_{0}$.
\item  $P(\sigma | \vec x_{n}) \geq 1-\epsilon_{1}$, then STOP and accept $H_1$, with guaranteed probability of success $s_{1}=1-\epsilon_{1}$.
\item If neither 1 nor 2 hold, continue sampling.
\end{enumerate}
Using \eqref{eq:cond0} and \eqref{eq:cond1}, the condition to continue sampling can be written in terms of a single sample statistic, the log-likelihood ratio 
\begin{align}
\label{eq:loglike}
Z_{n}=\log\frac{\Q(\vec{x}_n)}{\P(\vec{x}_n)}=\sum_{k=1}^{n} z_{k} \mbox{ with } z_{k}=\log\frac{\Q(x_{k})}{\P(x_{k})}
\end{align}
as  $b:=\log B\leq Z_{n}\leq \log A=:a$,  where $A=\frac{\eta_{0}}{\eta_{1}}\frac{1-\epsilon_{1}}{\epsilon_{1}} $, $B=\frac{\eta_{0}}{\eta_{1}}\frac{\epsilon_{0}}{1-\epsilon_{0}}$.

\begin{figure}[t]
\includegraphics[width=.8\columnwidth]{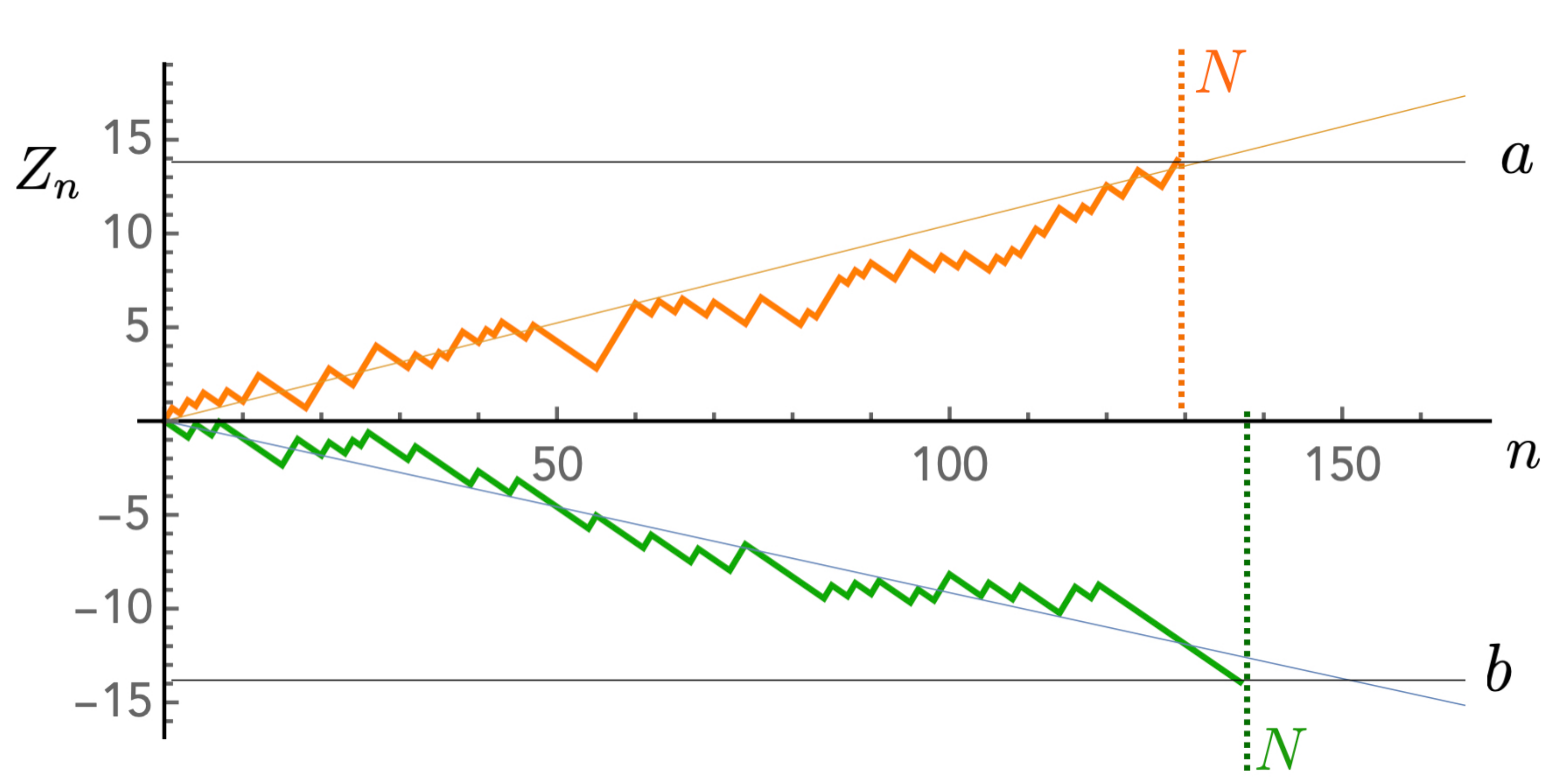}
\caption{Random walk describing the likelihood function $Z_{n}$ under $H_0$ (green) and $H_1$ (orange). When the value of $Z_n$ crosses $b\sim\log\epsilon_{0}$ ($a\sim\log\tfrac{1}{\epsilon_{1}}$) we decide in favour of $H_{0}$ ($H_{1}$). $N$ indicates the corresponding stopping time.}\label{fig:randomwalk}

\end{figure}

It is convenient to interpret $Z_{n}$  as a random walk (see Fig. \ref{fig:randomwalk}) that  at every instance performs a step of length $z_{k}$ with probability $p(x_{k})$, if $H_{0}$  holds, or with probability $q(x_{k})$, if $H_{1}$  holds. Under $H_{1}$ the mean position of the walker at step $n$ is given by $\mean{Z_{n}}_{1}=\sum_{k=1}^{n}\mean{z_{k}}_{1}=n \mean{z}_{1}=n D(q\| p)> 0$ where  $D( q\|  p)=\sum_{x}q(x)\log\frac{q(x)}{p(x)}$ is the relative entropy; while for $H_{0}$, \mbox{$\mean{Z_{n}}_{0}=-n D(p\|  q)<0$}. That is, under $H_{1}$ the walker has a drift towards  the positive axis, while under $H_{0}$ it drifts towards the negative axis. 
We define as the stopping time $N$ the first instance in which the walker steps out of the region $(a,b)$, i.e., $N:=\inf\{n : Z_{n}\notin (b,a)\}$, and note that it is a stochastic variable that \emph{only} depends on the current as well as the past measurement record.
The stochastic variable  $Z:=Z_{N}$ is the position of the walker at $N$. The mean value of this position can be related to the mean number of steps by Wald's identity~\cite{mitzenmacher2005probability},
\be\label{eq:Z1}
\mean{Z}_{1}=\mean{\sum_{k=1}^{N} z_{k}}_{1}=\mean{z}_{1}\mean{N}_{1}=D( q\|  p)\mean{N}_{1}
\ee
under hypothesis $H_{1}$, and likewise $\mean{Z}_{0}=-D( p\|  q)\mean{N}_{0}$.  
In order to estimate  $\mean{N}_{i}$ from \eqref{eq:Z1} we need to provide a good estimate for $\mean{Z}_{i}$. For this purpose let us first define $\mathcal{X}_{1}$ as the set of strings $\vec x$ such that $b<Z_{j}<a$ for all $j<n$  \emph{and}  $Z_{n}\geq a$, and $\mathcal{X}_{0}$ as the set of strings $\vec x$ such that $b<Z_{j}<a$ for all $j<n$  \emph{and}  $Z_{n}\leq b$.
 Then, the following relations hold:
\begin{align}\label{eq:alphaAbetaB}
\alpha=&P_{0}(Z\geq a)=\sum_{\vec x \in \mathcal{X}_{1}} p(\vec x) \leq 
\sum_{\vec x \in \mathcal{X}_{1}} \frac{q(\vec x)}{A}=\frac{1-\beta}{A}\\
\beta=&P_{1}(Z\leq b)=\sum_{\vec x \in \mathcal{X}_{0}}q(\vec x) \leq 
\sum_{\vec x \in \mathcal{X}_{0}} p(\vec x)B=(1-\alpha) B\nonumber
\end{align}
where in the first (second) inequality we used that  $\frac{q(\vec x)}{p(\vec x)}\geq A$ for strings in $\mathcal{X}_{1}$ ($\frac{q(\vec x)}{p(\vec x)}\leq B$ for strings in $\mathcal{X}_{2}$), and in the last equality we have used that $\lim_{n\to\infty}P(Z_{n}\in(b,a))=0$ \cite{Wald1948OptimumChar}, i.e. the walker eventually stops. The above equations are an instance of so-called Wald's likelihood ratio identity \cite{Lai_2004}.
We note that  the above inequalities can be taken to be approximate equalities if we assume that the process ends close to the prescribed boundary, i.e. there is no \emph{overshooting}. In particular, this will be valid in our asymptotically small error settings where the boundaries are far relative to  the (finite) step size $z_{k}$.
This allows us to establish a one-to-one correspondence between the thresholds $A,B$ and the type I \& II errors:
$\alpha \approx \frac{1 - B}{A - B}=
\frac{\epsilon_{1}(\eta_{1}-\epsilon_{0})}{( 1-\epsilon_{0} -\epsilon_{1}) \eta_{0}}$ and 
$\beta \approx \frac{B(A-1)}{A - B}=\frac{\epsilon_{0} (\eta_{0}-\epsilon_{1})}{( 1-\epsilon_{0} -\epsilon_{1}) \eta_{1}}$.
\cite{gendra_beating_2012,audenaert_eisert_2005}
Neglecting the overshooting also allows us to 
consider $Z$ as a stochastic variable that
takes two values $Z\in\{a,b\}$. Under hypothesis 
$H_{0}$, $a$ occurs with probability $P_{0}(Z=a)=\alpha$ 
and  $b$ with $P_{0}(Z=b)=1-\alpha $; while under hypothesis 
$H_{1}$, $a$ and $b$ occur with probabilities 
\mbox{$P_{1}(Z=a)=1-\beta$} and   $P_{1}(Z=b)=\beta$. So, 
 \begin{align}
\mean{Z}_{0}=a \alpha +b (1-\alpha) &\;{\rm and}&
\mean{Z}_{1}=a (1-\beta) +b \beta.
 \end{align}
Making use of \eqref{eq:Z1} one can  now write a closed expression for $\mean{N}_{0}$ and $\mean{N}_{1}$   in terms of $\epsilon_{1},\epsilon_{0}$ and the priors.
A remarkable property of the SPRT with  error probabilities $\alpha$ and $\beta$ is that it minimizes \emph{both} $\mean{N}_{0}$ and $\mean{N}_{1}$ among all tests (sequential or otherwise) with bounded  type I and type II  errors.
 This optimality result due to Wald and Wolfowitz \cite{Wald1948OptimumChar} allows us to extend the above results to the asymmetric scenario. For the symmetric scenario, the SPRT has also been shown \cite{simons_improved_1976} to be optimal among all tests respecting a bounded mean error $\bar\epsilon'\leq \bar\epsilon$.
In the asymptotic limit of small error bounds, $\epsilon_{0},\epsilon_{1}\ll 1$, the  threshold values are $a~\sim-\log \epsilon_{1}$ and $b~\sim\log \epsilon_{0}$, which correspond to $\alpha\sim \frac{\eta_{1}}{\eta_{0}}\epsilon_{1}$ and $\beta\sim \frac{\eta_{0}}{\eta_{1}}\epsilon_{0}$, yielding
 \be
\label{eq:Ni}
\mean{N}_{0}\sim-\frac{\log\epsilon_{0}}{D( p\|  q)} \;\;\mathrm{ and }\;\;
\mean{N}_{1}\sim-\frac{\log\epsilon_{1}}{D( q\|  p)}.
\ee
The same expressions hold at leading order in the asymmetric scenario when the type I \& II errors are vanishingly small, replacing $\log\epsilon_1$ and  $\log\epsilon_0$ by $\log\alpha$ and $\log\beta$ respectively---and in the symmetric scenario
replacing both quantities by $\log\bar\epsilon$.
If one of the error thresholds,  say $\alpha$, is kept finite while the second is made vanishingly small $\beta \ll 1$,  $\mean{N}_{1}$ remains finite, while the other conditional mean scales as $\mean{N}_{0}\sim-\frac{(1- \alpha)\log\beta}{D( p\|  q)}$.

In the supplemental material \cite{SM} we apply these results to the discrimination of qubit states using projective measurements and give closed expressions for the optimal Bayesian mean number of copies $\mean{N}:=\eta_0\mean{N}_0+\eta_1\mean{N}_1$.  Figure \ref{fig:mix} shows that in the symmetric setting these restricted sequential strategies already require on average between 25-50\% less resources than the best deterministic  strategy  that uses a fixed number  of copies $N_{\rm Ch}\sim-\log\bar\epsilon/\xi_{\rm Ch}$  \cite{calsamiglia_quantum_2008,audenaert_discriminating_2007}, 
and  requires non-trivial collective measurements \cite{calsamiglia_local_2010}.

\begin{figure}[t]
\centering
      \includegraphics[width=1.\columnwidth]{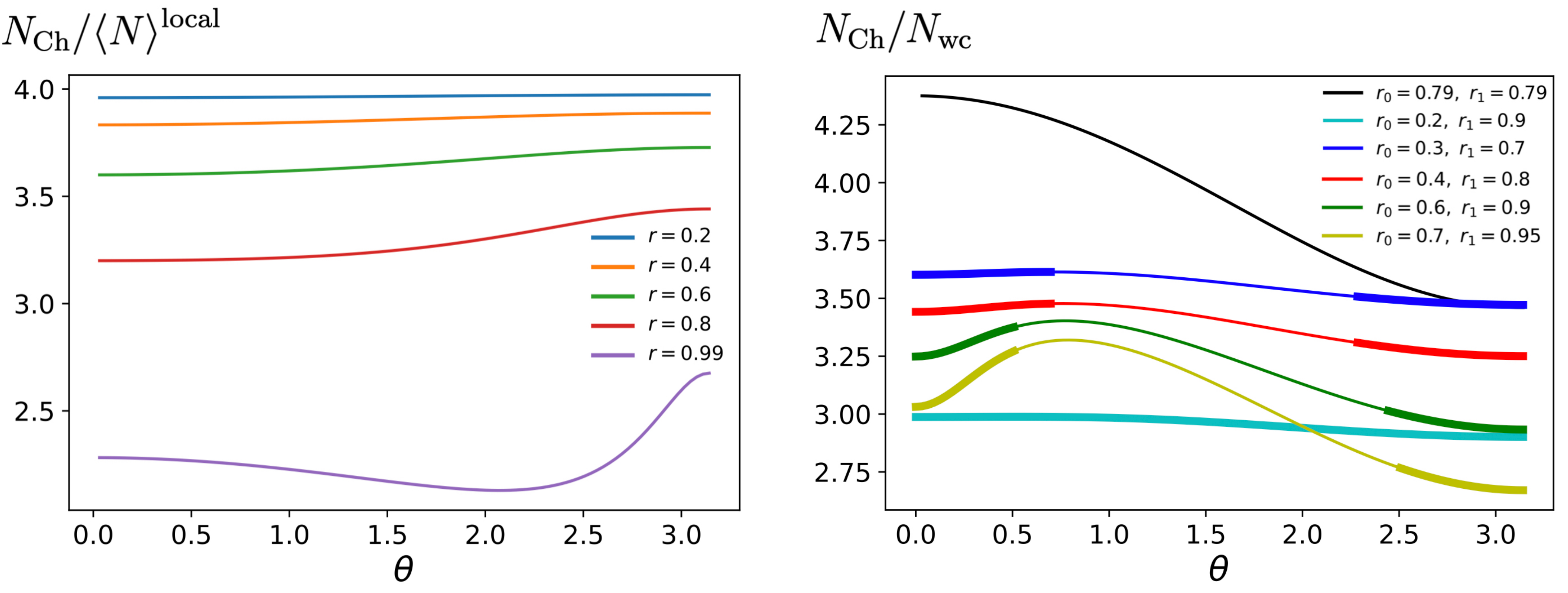}
        \caption{
		Left: ratio between the number of copies required by the best deterministic strategy, $N_\text{Ch}$, and the Bayesian mean number of copies for a sequential strategy based on fixed local unbiased measurements, $\mean{N}_\text{local}$, for pairs of states of purity $r$ and relative angle $\theta$ \cite{SM}. Right: ratio between $N_\text{Ch}$ and the worst-case $N_\text{wc}$, for pairs of states of different purities $r_0$, $r_1$. The thin lines use the expression~\eqref{lowerbound} for $N_\text{wc}$, whereas the thick lines represent the cases for which this ultimate limit of $N_{\text{wc}}$ is attained by a block-sampling strategy.
		}
    \label{fig:mix}
\end{figure}

\paragraph{Ultimate quantum limit.}
Quantum mechanics allows for much more sophisticated strategies. For a start, performing a non-projective generalized measurement already gives important advantages (see below). One can also adapt the measurements depending on the previous measurement outcomes and, importantly,  measurements may be weak so that each new measurement  acts on a fresh copy but also on the preceding, already measured, copies.
Without loss of generality we can assume that at every step $k$ we perform a measurement with three outcomes $x_{k}\in\{0,1,2\}$: the first two must fulfill  conditions \eqref{eq:cond0} and \eqref{eq:cond1} and trigger the corresponding guess ($H_{0}$ or $H_{1}$ respectively), while the third outcome signals to continue measuring having an additional fresh copy available. 
The measurement at step $k$ is characterized by a quantum instrument
$\mathcal{M}^{k}=\{\mathcal{M}^{k}_{0},\mathcal{M}^{k}_{1},\mathcal{M}^{k}_{2}\}$, and the  sequential strategy is given by a sequence of instruments $\mathcal{M}=\{\mathcal{M}^{k}\}_{k=1}^{\infty}$. With this,  given hypothesis $\nu=\{0,1\}$, the probability of getting outcome  $x_{k}$ at step $k$  can be written as
\be
P_{\nu}(x_{k}):= \tr[ \mathcal{M}^{k}_{x_{k}}\circ \mathcal{M}^{k-1}_{2}\ldots \circ \mathcal{M}^{1}_{2}(\rho^{\otimes k}_{\nu})]=\tr(E_{x_{k}}^k \rho^{\otimes k}_{\nu})
\label{eq:probk}
\ee
where we have used that in order to arrive to step $k$ a ``continue'' outcome must be triggered in all previous steps, and in the last equality we have  defined the effective POVM  $E^{k}=\{E_{i}^{k}\}_{i=0}^{2}$.
Making use of the indicator function $\id_{k\leq N}$, the mean number of steps under hypothesis $\nu$ can be computed as
\be
\mean{N}_{\nu}=\mean{\sum_{k=1}^{N} 1}_{\nu}=
\mean{\sum_{k=1}^{\infty} \id_{k\leq N}}_{\nu}=
\mean{\sum_{n=0}^{\infty} \id_{n< N}}_{\nu}=
\sum_{n=0}^{\infty} T_{\nu}^{n}
\label{eq:meanNnu}
\ee  
where  $T_{\nu}^{n}=P_{\nu}(n<N)$ is the probability that the sequence does not stop at step $n$, which from \eqref{eq:probk} is given by $T^{n}_{\nu}= P_{\nu}(x_{n}=2)$.
Optimizing $\mean{N}_{\nu}$ over all quantum sequential strategies 
$\mathcal{M}$ is daunting, as all terms $T_{\nu}^{n}$  are  strongly interrelated through the intricate structure of $E^{n}$. However,  a lower bound  to each $T^{n}_{\nu}$ can be found by relaxing such structure and only imposing minimal requirements on the effective POVM; namely the error bounds 
\eqref{eq:alphaAbetaB},  positivity and  completeness:
\begin{align}\label{optimQn}
&\min_{E^{n}} \tr(E_{2}^n \rho^{\otimes n}_{\nu}) \; {\rm s.t.} \; E^n_i\geq 0,\; \sum_{i=0}^{2} E^n_i=\id, 
\;\mathrm{and}  \\
&  \tr[E^n_1\left( \sigma^{\otimes n} -A \rho^{\otimes n}\right)]\geq 0 ,\; 
\tr[E^n_0\left( \rho^{\otimes n} -B^{-1} \sigma^{\otimes n}\right) ]\geq 0. \nonumber
\end{align}
This semi-definite program, which can be considered a two-sided version of the quantum Neyman-Pearson test \cite{audenaert_asymptotic_2007}, is an interesting open problem in its own right. Our focus, however, is the asymptotic regime of small error bounds.
In these asymptotic scenarios we are able to show, exploiting some recent strong converse results in hypothesis testing \cite{cooney_strong_2016,beigi_quantum_2020}, that for all  {$n<n^*=-\frac{\log\epsilon_{0}(1-A^{-1})}{D(\rho\| \sigma)} $},  $T^{n}_0 \geq 1-O(\epsilon_{0}^{\kappa})$ for some  $\kappa \in(0,1)$ \cite{SM},
which leads to the desired bound:
\begin{equation}
    \label{lowerbound}
\mean{N}_0\geq \sum_{n=0}^{\lfloor n^{*}\rfloor} T_{0}^{n} \geq  -\frac{\log\epsilon_{0}(1-A^{-1})}{D(\rho\| \sigma)} +O(1) \,.
\end{equation}
An analogous bound holds for $\mean{N}_1$. The bounds for asymmetric (symmetric) scenarios (see \cite{SM}) take the same form, replacing $\log\epsilon_{0}$ by $\log\beta$ ($\log\bar\epsilon$) and $A^{-1}$ by $\alpha$  ($\bar\epsilon$). 
In the asymmetric scenario where $\epsilon_1$ or $\alpha$ is kept finite, it also holds that $\mean{N}_1=O(1)$ and $\mean{N}_0$ is given by the appropriate version of \eqref{lowerbound}.

\paragraph{Attainability and upper bounds.} 
Consider a sequential strategy that involves a fixed, collective measurement 
${K}=\{K_{i}\}$, acting on consecutive blocks of $\ell$ copies, 
 yielding two possible  distributions $\P^{\ell}_K, \Q^{\ell}_K$. 
 Using the classical SPRT we get that   \be
 \mean{N}_{0}= \ell \inf_{K} \mean{M}_{0}\sim\ell\inf_{K} \frac{-\log\epsilon_{0}}{D(\P^{\ell}_K\| \Q^{\ell}_K)}\sim \frac{-\log\epsilon_{0}}{D(\rho\| \sigma)}
 \label{batching}
 \ee
 where $M$ is the number of blocks used at the stopping time.   In the last relation of \eqref{batching} we have used the fact that we are in the asymptotic setting where 
 $\epsilon_0 \ll 1$
 and therefore we can take arbitrarily long block lengths $\ell\gg 1$.  We also exploit the following  property of the measured relative entropy \cite{hiai_proper_1991,Hayashi_2001}: $\sup_{K} D(\Q^{\ell}_K\| \P^{\ell}_K)\sim\ell D(\sigma\| \rho)$.
 
 Notice, however, that for arbitrary states $\rho$ and $\sigma$ block sampling can attain either $\mean{N}_0$ or $\mean{N}_1$, but it is unknown whether one can attain in general both bounds simultaneously, i.e., whether a measurement achieving the supremum of $\lim_{\ell\to \infty}\tfrac{1}{\ell}D(\Q^{\ell}_K\| \P^{\ell}_K)$ can also attain the supremum of $\lim_{\ell\to \infty}\tfrac{1}{\ell}D(\P^{\ell}_K\| \Q^{\ell}_K)$. For instance, if we wish to optimize the Bayesian mean number of copies $\mean N$,
 we can use block sampling to attain
\begin{equation}
 \mean{N}_{\mathrm{block}} \sim \lim_{\ell\to\infty} \inf_{K} \left(\frac{-\ell\eta_0\log\epsilon_{0}}{D(\P_K^\ell\| 
\Q_K^\ell)}-\frac{\ell \eta_1\log\epsilon_{1}}{D(\Q_K^\ell\| \P_K^{\ell})}\right).
\label{eq:Nblock}
\end{equation}
However,  this  strategy might be sub-optimal and hence it only provides an upper bound to the optimal Bayesian mean $\mean N \leq \mean{N}_{\mathrm{block}}$. 
This notwithstanding, there are at least two cases when this upper bound coincides with the lower bound provided by \eqref{lowerbound}: when $\rho$ and $\sigma$ commute, and when
the two states do not have common support. If, say, $\operatorname{supp}(\sigma) \cap \operatorname{ker}(\rho) \neq 0$,  one can use block-sampling to attain \eqref{lowerbound} for $\mean{N}_{0}$ and always detect $\rho$ with a finite number of copies---note that since $D(\sigma\| \rho)=\infty$, the lower bound $\mean{N}_{1}=O(1)$ is also attained. 

We can also give achievable lower bounds for a worst-case type figure of merit
${N}_{\rm wc}:=\max\{\mean{N}_0,\mean{N}_1\}$.
If, say, $\mean{N}_0>\mean{N}_1$, then in 
\cite{SM} we give some instances of qubit pairs
where a specific block-sampling strategy 
\cite{Hayashi_2001} saturates \eqref{lowerbound} for 
$\mean{N}_0$, while at the same time 
$\lim_{\ell\to\infty}\tfrac{1}{\ell}D(\Q_K^\ell\| 
\P_K^\ell)\geq D(\rho \| \sigma)$, and hence  
\eqref{lowerbound} provides the ultimate attainable 
limit for ${N}_{\rm wc}$. 
In Fig.~\ref{fig:mix} we compare ${N}_{\rm wc}$ with $N_{\text{Ch}}$ for several pairs of states, highlighting the achievable cases, and show a consistent advantage of sequential protocols over deterministic ones~\footnote{Note that the comparison with $N_{\text{Ch}}$ is unfavorable to sequential strategies. Substituting $\epsilon_0$ and $\epsilon_1$ by $\bar{\epsilon}$ in Eq.~\eqref{lowerbound} implies that \emph{each} type of error is independently constrained, whereas $N_\text{Ch}$ refers to a deterministic (symmetric) protocol where the \emph{mean} error is $\bar{\epsilon}$ and thus to a weaker version of the problem. In spite of this, the sequential scenario displays a significant advantage.}.

Finally, we note that, in an asymmetric scenario where $\mean{N}_1$ is finite and the value of $\mean{N}_0$ achieves the lower bound~\eqref{lowerbound}, sequential protocols provide a strict advantage over Stein's limit for deterministic protocols by a factor $(1-\alpha)$.

\paragraph{The curious case of pure states.}  
If the two states are pure, the behavior of $\mean{N}_{\nu}$ changes drastically: it is possible to reach a decision with guaranteed zero error using a finite average number of copies. 
To see this, consider again Eq.~\eqref{eq:meanNnu}. Under a zero-error condition, the minimal (unrestricted) $T_\nu^{n}$ is achieved by a global \emph{unambiguous} three-outcome POVM \cite{Ivanovic1987,Dieks1988,peres_how_1988} on $n$ copies, which identifies the true state with zero error when the first or the second outcome occurs ---at the expense of having a third, inconclusive outcome. 
For a single-copy POVM over pure states, the probabilities $c_\nu$ of the inconclusive outcome under $H_\nu$ are subject to the tradeoff relation $c_0 c_1 \geq \tr\rho\sigma$~\cite{Sentis2018}, where equality can always be attained by suitable POVM that maximizes the probability of a successful identification. Likewise, for a global measurement on $n$ copies we have $T_0^n T_1^n\geq (\tr\rho\sigma)^n$. Now, it is evident that a sequence of $n$ locally optimal unambiguous POVMs applied on every copy, for which $T_\nu^n=c_\nu^n$, also fulfills the global optimality condition. Hence, we have
\begin{equation}
\mean{N}_\nu \geq \sum_{n=0}^\infty T_\nu^n = \sum_{n=0}^\infty c_\nu^n = \frac{1}{1-c_\nu} =: \mean{N}_\nu^{\rm local} \,.
\label{eq:unamb}
\end{equation}
This shows that, for pure states, it suffices to perform local unambiguous measurements to attain the optimal (finite) average number of copies with zero error under hypothesis $H_\nu$. Note that because of the tradeoff $c_0 c_1\geq \tr\rho\sigma$ one cannot attain the minimal values of $\mean{N}_0^{\rm local}$ and $\mean{N}_1^{\rm local}$ for general states $\rho, \sigma$, simultaneously.
%
%
For instance,  one can reach the minimal value $c_{0}=\tr(\rho\sigma)$ for one hypothesis, but then having a maximal value $c_{1}=1$ for the second; or choose the optimal symmetric setting,  $c_{0}=c_{1}=\sqrt{\tr\rho\sigma}$, that 
achieves the minimum value of
both the worst-case $N_{\mathrm wc}$ and the Bayesian mean $\mean{N}$ with equal priors (see \cite{SM}). 
This is in stark contrast with the behavior found in \cite{Slussarenko2017MinASN}, where all strategies considered were based on two-outcome projective measurements, for which the average number of copies scaled as $\mean{N}\propto -\log\epsilon$.

\paragraph{Acknowledgments.}
We acknowledge financial support from the Spanish Agencia Estatal de Investigaci\'on, project 
PID2019-107609GB-I00,  from Secretaria d'Universitats i Recerca del Departament d'Empresa i Coneixement de la Generalitat de Catalunya, co-funded by the European Union Regional Development Fund within the ERDF Operational Program of Catalunya (project QuantumCat, ref. 001-P-001644), and Generalitat de Catalunya CIRIT 2017-SGR-1127. CH acknowledge financial support from the European Research Council (ERC Grant Agreement No. 81876), VILLUM FONDEN via the QMATH Centre of Excellence (Grant No.10059) and the QuantERA ERA-NET Cofund in Quantum Technologies implemented within the European Union Horizon 2020 Programme (QuantAlgo project) via Innovation Fund Denmark, JC acknowledges from the QuantERA grant C'MON-QSENS!, via Spanish MICINN PCI2019-111869-2, EMV thanks financial support from CONACYT.
JC also acknowledges support from ICREA Academia award.

\bibliographystyle{apsrev4-1}
%

\begin{widetext}
\section{Supplemental Material}

\maketitle

This supplemental material contains some technical details as well as  some extensions for the interested reader. 
In  Sec.~\ref{sec:Qstar} we state and prove  a theorem that provides the  lower bounds for the average number of copies,  Eq.~(12) in the main text.  In Sec.~\ref{sec:qubits} explicit
expressions for the optimal sequential test and attainability regions of the worst case bound are provided for the qubit case. In Sec.~\ref{sec:finitedim}  we present a second
theorem that provides  a general lower bound for the deviation of the measured entropy from its maximum value for arbitrary finite dimensions.  Finally, in Sec.~\ref{sec:unambiguous} we give the optimality 
proof for the zero-error protocol for pure states.

\section{Converse proof}
\label{sec:Qstar}
Our aim here is to prove that if one of the two error bounds is vanishingly small, say $\epsilon_{0}\ll 1$ under the strong error condition, or $\beta\ll 1$ in the asymmetric scenario (see main text for the definitions), then  the mean number of copies under hypothesis $H_{0}$ is always lower-bounded by
\begin{align}
    \label{lower bound}
\mean{N}_0&\geq \sum_{n=0}^{\lfloor n^{*}\rfloor} T_{0}^{n} \geq  -\frac{\log\epsilon_{0}(1-1/A)}{D(\rho\| \sigma)} +O(1) \,\quad \mbox{ or }\\
\mean{N}_0&\geq \sum_{n=0}^{\lfloor n^{*}\rfloor} T_{0}^{n} \geq  -\frac{\log\beta(1-\alpha)}{D(\rho\| \sigma)} +O(1) \,,
\end{align}
respectively. Analogous bounds also hold for $\mean{N}_{1}$ when $\epsilon\ll 1$ or $\alpha\ll 1$, replacing $\epsilon_{0}\leftrightarrow \epsilon_{1}$, $A\leftrightarrow B$, $\alpha\leftrightarrow \beta$, and $\rho\leftrightarrow \sigma$.
We will first provide the proof for the strong error condition and indicate how to adapt it to  the other hypothesis testing scenarios considered here.

In the main text (MT) we have shown that under hypothesis $H_{0}$ the mean number of sampled copies  is given by
\be
\mean{N}_{0}=\sum_{n=0}^{\infty} T_{0}^{n}\geq \sum_{n=0}^{n^{*}} T_{0}^{n}	\,,
\label{eq:meanNnu}
\ee  
where the last inequality holds for all values of $n^{*}$ since $T_{0}^{n}\geq0$. The $n$th term in sum, $T_{0}^{n}=P_{0}(n<N)$, is the probability of getting a ``continue'' outcome at step $n$, corresponding to the POVM element $E_{2}^{n}$ implicitly defined in (12) 
in MT. The continue probability at a particular step $n$ obeys a lower bound  given by the following semidefinite program (SDP)
\be\label{eq:Tks}
T_{0}^{n}\geq\tilde T_{0}^{n}:=\min_{E^{n}} \tr(E_{2}^n \rho^{\otimes n}_{\nu}) \; {\rm s.t.}
\left\{
\begin{array}{ll}
& \mathrm{0) }\quad \{E^n_i\geq 0\}_{i=0}^{2} \;\mathrm{ and } \sum_{i=0}^{2} E^n_i=\id \\\\
& \mathrm{1) }\quad  \tr[E^n_1\left( \sigma^{\otimes n} -A \rho^{\otimes n}\right)]\geq 0 \; \\\\
&\mathrm{2) }\quad  \tr[E^n_0\left( \rho^{\otimes n} -B^{-1} \sigma^{\otimes n}\right) ]\geq 0
\end{array}\right. .
\ee
The  conditions in 0) have to hold for any valid POVM, while the second and third conditions are an alternative way of writing the strong errors conditions, (1) and (2) in the MT, as SDP  constrains.
For small error bounds $\epsilon_{0}\ll 1$ (i.e. $A=\frac{\eta_{0}}{\eta_{1}}\frac{1-\epsilon_{1}}{\epsilon_{1}}\gg 0$) the solution of the SDP program in \eqref{eq:Tks} has a  characteristic dependence on $n$ as illustrated in Figure \ref{fig:Tvskplot}.
\begin{figure}[h]
    \includegraphics[scale=0.35]{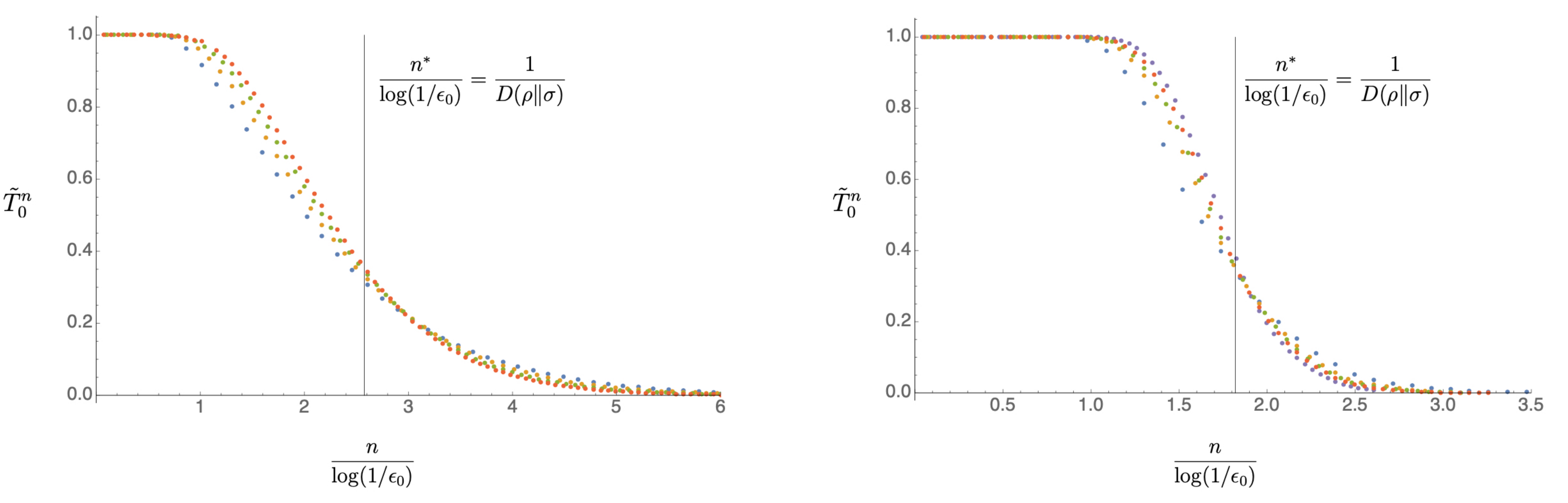}
    \caption{Lower bound $\tilde T_{0}^{n}$ on the continue probability as a function of the step number $n$ for: (Left) two  qubits of equal purity $r=0.9$ and relative angle $\theta=\pi/4$, for error bounds $\epsilon_{0}=\{10^{-3},10^{-4},10^{-5},10^{-6}\}$ (from bottom to top, looking at the left side of the figure); (Right) two commuting qubits with $r=0.5$ and $\theta=\pi$, for $\epsilon_{0}=\{10^{-8},10^{-12},10^{-14},10^{-16},10^{-20}\}$ (from bottom to top, looking at the left side of the figure). The values on the left plot have been obtained by numerically solving the SDP program of \eqref{eq:Tks}, exploiting the block-diagonal structure of iid quantum states  $\rho^{\otimes n}$ and $\sigma^{\otimes n}$ (see Section \ref{sec:qubits}).}
\label{fig:Tvskplot}
\end{figure}
When the number of sampled copies $n$ is small  it is not possible to meet the low error bound and the probability of  getting a ``continue'' outcome is  $T_{0}^{n}=1$. This probability  remains constant as $n$ increases until it approaches the critical point $n^{*}\sim \log(1/\epsilon) /D(\rho \|  \sigma)$, at which point it rapidly drops to zero. Note that this drop becomes more abrupt as  $\epsilon_{0}$ decreases. 
These observations suggest that $\sum_{n=0}^{n^{*}} \tilde T^{n}_{0}$ is a very tight lower bound to  $\sum_{n=0}^{\infty} \tilde T^{n}_{0}$ (area under the curve in  Figure \ref{fig:Tvskplot}) and is given to a very good approximation by $\sum_{n=0}^{n^{*}} \tilde T^{n}_{0}\approx n^{*}$.

With this at hand we can now carry on with the formal presentation and proof of the lower bound.

\begin{theorem}
Given two finite-dimensional states, $\rho$  ($H_{0}$) and $\sigma$ ($H_{1}$), occurring with prior probabilities
$\eta_{0}$ and $\eta_{1}$ respectively, the most general quantum sequential strategy that satisfies  the strong error conditions $P(H_{0}|x_{N}=0)\geq 1-\epsilon_{0}$ and $P(H_{1} |x_{N}=1)\geq 1-\epsilon_{1}$, where  
$x_{N}\in\{0,1\}$ is the output of the measurement at the stopping time $N$, necessarily fulfills the following asymptotic lower bound for the mean number of sampled copies when $\epsilon_{0}\ll 1$:
    \begin{equation}
\mean{N}_0\geq   -\frac{(1-A^{-1})\log\epsilon_{0}}{D(\rho\| \sigma)} +O(1) \,, \quad \mbox{ where }\quad \,A=\frac{\eta_{0}}{\eta_{1}}\frac{1-\epsilon_{1}}{\epsilon_{1}}.
\label{eq:N0low1}
\end{equation}

Similarly, the most general quantum sequential strategy that satisfies  the (weak)  error conditions $P(\hat H_{1}|\rho)\leq \alpha$ and $P(\hat H_{0}|\sigma)\leq \beta$, where $\hat H_{0}$ and $\hat H_{1}$ are the events of accepting hypothesis 0 and 1 respectively, necessarily fulfills the following asymptotic lower bound for the mean  number of sampled copies when $\beta\ll 1$:
    \begin{equation}
\mean{N}_0\geq   -\frac{(1-\alpha)\log\beta}{D(\rho\| \sigma)} +O(1) \,.
\label{eq:N0low2}
\end{equation}
\end{theorem}

\begin{proof}

We start by noting that the strong error conditions [see \eqref{eq:Tks}] imply
\begin{eqnarray}
\label{eq:E1bounds}
\tr (E^n_1\rho^{\otimes n})&\leq&\frac{1}{A} \tr (E^n_1 \sigma^{\otimes n})\leq 
\frac{1}{A}=\frac{\eta_{1}}{\eta_{0}}\frac{\epsilon_{1}}{1-\epsilon_{1}} \,,\\
\tr (E^n_0 \sigma^{\otimes n})&\leq& \frac{1}{B} \tr (E^n_0\rho^{\otimes n})\leq
\frac{1}{B}=\frac{\eta_{0}}{\eta_{1}}\frac{\epsilon_{0}}{1-\epsilon_{0}} \,.
\label{eq:E0bounds}
\end{eqnarray}

Next we form a two-outcome POVM by binning two outcomes of the effective POVM at step $n$, as defined in (12) in MT, \mbox{$ F^{n}=\{ F^n_0=E^{n}_{0},{F}^n_1:= E^n_1+E^n_{2}\}$}. This measurement can be used to  discriminate between $\rho^{\otimes n}$ and $\sigma^{\otimes n}$ and the associated type-I and type-II errors will be denoted by $\tilde{\alpha}_n=\Tr[F^n_1\rho^{\otimes n}]$ and
$\tilde{\beta}_n=\Tr[F^n_0\sigma^{\otimes n}]$.
From \eqref{eq:E0bounds} and the above definitions it follows that
$\tilde{\beta}_n\leq \frac{\eta_{0}}{\eta_{1}}\epsilon_{0}+O(\epsilon_{0}^{2})$.
In addition, using \eqref{eq:E0bounds} we find that the probability of continuing at step $n$
when $H_{0}$ holds satisfies 
\be
T_{0}^{n}=\tr[E^n_2\rho^{\otimes n}]=\tilde{\alpha}_n-\tr[E^n_1\rho^{\otimes n}]\geq \tilde{\alpha}_n-\frac{1}{A} \,.
\label{eq:T0n}
\ee

Now we use Lemma 1 stated below, which uses the recently developed methods for strong converse exponents \cite{cooney_strong_2016} in order to establish a lower bound  on $\tilde{\alpha}_n$ 
when the type-II error $\tilde{\beta}_n$ is bounded by $\epsilon$  and when $n$ is below a critical threshold $n^{*}$.
In particular, applying  Lemma 1 to the test defined by $F_{n}$ above, with
type-I \& II errors $\tilde \alpha_{n}$ and $\tilde{\beta}_n\leq \frac{\eta_{0}}{\eta_{1}}\epsilon_{0}+O(\epsilon_{0}^{2})=:\epsilon$, we have that \eqref{eq:T0n} reads
\be
T^{n}_{0}\geq  \tilde \alpha_{n}-\frac{1}{A}\geq 1-\epsilon^{ \kappa_n(\rho,\sigma)} -\frac{1}{A} 
\quad  \forall \; n<n^{*} \,,
\ee
where $\kappa_n(\rho,\sigma)>0$.
The proof for the strong error conditions ends by inserting this lower bound in \eqref{eq:meanNnu}.

For the weak form of error bounds one can follow the same steps as above by writing the type-I and type-II errors of the sequential strategy as $P(\hat H_{1}|\rho)=\sum_{k=1}^{\infty} \tilde \alpha_{n}$ and $P(\hat H_{0}|\sigma)=\sum_{k=1}^{\infty} \tilde \beta_{n}$. Since $\tilde \beta_{n}\geq 0$,  the error bound $P(\hat H_{0}|\sigma)\leq \beta$  translates to 
$ \tilde \beta_{n}\leq \beta$, and similarly $ \tilde \alpha_{n}\leq \alpha$.  The former is directly of the form required for Lemma 1, while the latter can be used instead of \eqref{eq:E1bounds}, i.e., $\tr (E^n_1\rho^{\otimes n})=\alpha_{n}\leq\alpha$, hence  $A$ in \eqref{eq:N0low1} becomes $\alpha$ in  \eqref{eq:N0low2}.

\end{proof}

\begin{lemma}
Let $\rho$ and $\sigma$ be finite-dimensional density operators associated to hypotheses $H_{0}$ and $H_{1}$, respectively. For any quantum hypothesis testing strategy that uses $n$ copies of the states and that respects the type-II error bound $\beta_n \leq \epsilon$, with $\epsilon\ll 1$, the type-I error will converge to one at least as  
    \begin{equation}
        \alpha_n \geq 1-\epsilon^{\kappa_n(\rho,\sigma)}  \quad\mbox{for all}\quad n < n^{*}=\frac{-\log\epsilon}{D(\rho \| \sigma)} \,,
    \end{equation}
    where $0<\kappa_n(\rho,\sigma)<1$ is given by
    \begin{equation}
        \kappa_n(\rho,\sigma)=\sup_{s>1}\frac{s-1}{s}\frac{\xi_n-\tilde{D}_s(\rho\| \sigma)}{\xi_n} = \frac{H(\xi_n)}{\xi_n}\,,
         \quad  \mathrm{ with }\quad   \xi_n=-\frac{ \log\epsilon}{n}>-\frac{ \log\epsilon}{n^{*}}=D(\rho\| \sigma) \,,
    \end{equation}
    where $H(\xi_n)$ is the strong converse exponent~\citep{Mosonyi2015} and where  the sandwiched Renyi relative entropy~\cite{Muller2013OnQuantum,wilde_strong_2014} is given by
\be
\tilde{D}_s(\rho\| \sigma)=\frac{1}{s-1} \log \tr\!\left(\sigma^{\frac{1-s}{2 s}} \rho 
\sigma^{\frac{1-s}{2 s}}\right)^{s} \,,
\ee   
taking $\tilde{D}_s(\rho\| \sigma) =\infty$ when $\operatorname{supp} \rho \not\subseteq \operatorname{supp} \sigma$.

\label{eq:kappa}
\end{lemma}
\begin{proof}
The proof makes use of the following strong converse result by Mosonyi and Ogawa ~\cite{Mosonyi2015} that relates the type I and type II errors for an arbitary $n$ by means of the sandwiched Renyi relative entropy:
\begin{equation}
    \frac{1}{n}\log(1-{\alpha_n})\leq\frac{s-1}{s}(\tilde{D}_s(\rho\| \sigma)+\frac{1}{n}\log\beta_n),~~s>1 \,.
\end{equation}
Note that in order to avoid confusion with the type-I error, here we use  $s$ instead of the traditional $\alpha$ used in the Renyi entropies. Among the number of properties that make the sandwiched Renyi relative entropy such a  formidable quantity, here we will use two: i) it increases monotonically with $s$, and ii) $\lim_{s\to 1} \tilde{D}_s(\rho\| \sigma)=D(\rho\| \sigma)$.

Since $\beta_n \leq \epsilon$,
\begin{equation}
    1-{\alpha_n}\leq  \ex{n\frac{s-1}{s}(\tilde{D}_s(\rho\| \sigma)+\frac{1}{n}\log\epsilon)}\,.
    \label{eq:alph}
\end{equation}
Observe that $\forall n < n^{*}$ we can define $\xi_{n}>D(\rho\| \sigma)$ such that
\begin{equation}
    n=-\frac{\log\epsilon}{\xi_n} \,.
    \label{eq:deltatrick}
\end{equation}
Using this parametrization of $n$ in \eqref{eq:alph}, we have
\begin{equation}
    {\alpha_n}\geq 1 -\epsilon^{\frac{s-1}{s}\frac{\xi_n-\tilde{D}_s(\rho \| \sigma)}{\xi_{n} }}\,.
\end{equation}
Hence, if we define the supremum of the exponent
\begin{equation}
        \kappa_n(\rho,\sigma):=\sup_{s>1}\frac{s-1}{s}\frac{\xi_n-\tilde{D}_s(\rho\| \sigma)}{\xi_n} \,,
         \quad  \mathrm{ with }\quad   \xi_n=-\frac{ \log\epsilon}{n}>-\frac{ \log\epsilon}{n^{*}}=D(\rho\| \sigma)\,,
         \label{eq:kappa1}
 \end{equation}
we arrive to the desired result
\begin{equation}
    \alpha_n\geq 1-\epsilon^{ \kappa_n(\rho,\sigma)}~\quad \forall~n< n^*\,.
\end{equation}

%
Taking into account that $0<\frac{s-1}{s}=1-\frac{1}{s}<1$ for all $s>1$, that $\xi_n>D(\rho \| \sigma)$ for $n<n^*$, and from conditions i) and ii) above that $\tilde{D}_s(\rho\| \sigma)>D(\rho\| \sigma)$, 
it follows that there will always be an $s'$ realizing the supremum in~\eqref{eq:kappa1} such that $\xi_n > \tilde{D}_{s'}(\rho\|\sigma)$, and therefore $0<\kappa_n(\rho,\sigma)<1$.
\end{proof}

An alternative way to arrive to the result in  Lemma 1 is provided  in Beigi {\it et al.} \cite{beigi_quantum_2020} where, using quantum reverse hypercontractivity, a second order strong converse result on hypothesis testing is derived.

We finally note that Lemma 1 assures that, below $n^{*}$, the continue probability is $T^{n}\sim 1$. On the other hand, from Stein's Lemma we know that, for fixed (large) $n$, the optimal type-II error rate is given by the relative entropy, i.e., $\beta_n \sim \ex{-n D(\rho\| \sigma)}$. This explains why one does not need to continue measuring after $n>n^{*}=-\log\epsilon/D(\rho\| \sigma)$, and  $T^{n>n^{*}}\sim 0$ (see Fig.~\ref{fig:Tvskplot}), and why we may expect the lower bound to be tight, in the sense that we are not dropping significant contributions by truncanting the sum in \eqref{eq:meanNnu}.
 Of course, this still does not imply the attainability of the lower bound, and even less the simultaneous attainability of the bound for $\mean{N}_{0}$ and the analogous bound for $\mean{N}_{1}$.
 
 \pagebreak
\section{Sequential hypothesis testing for Qubits}\label{sec:qubits}
In this section we study the discrimination of qubit states using 
sequential methodologies, deriving explicit formulae for the mean 
number of copies using different measurement strategies.

\subsection{Optimal sequential test for fixed projective measurements}
We will first study the optimal performance under the simplest type of measurement apparatus, i.e. a fixed Stern-Gerlach-type measurement. The main purpose of this section is to show that using sequential strategies a simple projective measurement can determine the correct hypothesis with guaranteed bounded error requiring an expected number of copies significantly lower than the most general collective measurement acting on a fixed number of copies. In addition, we provide closed expressions for the optimal asymptotic performance.

Without loss of generality we characterize the two hypotheses by
\begin{equation}
\begin{split}
\rho &= r_{0} \ketbrad{\psi_0}+(1-r_{0})\id/2\\
\sigma&=r_{1} \ketbrad{\psi_1}+(1-r_{1})\id/2\,, 
\end{split}
\label{eq:mixedcase}
\end{equation}
where $\ket{\psi_i} = \cos\frac{\theta}{4}\ket{0}+(-1)^i \sin\frac{\theta}{4}\ket{1}$, $0\leq\theta\leq \pi$, $0\leq r_{i}\leq 1$ and  the (fixed) local measurement as
 $E_{0}=\ketbrad{\phi}$ and $E_{1}=\id-E_{0}$, with  
 $\ket{\phi}=\cos{\frac{\phi}{2}}\ket{0}+ \sin{\frac{\phi}{2}}\ket{1}$, and $0\leq\phi\leq \pi$. With these parametrizations, the probabilities of obtaining outcome $i=0,1$ are 
  $p_{\phi}(i)=P(i|H_0)=\frac{1}{2}[1+(-1)^{i}\cos({\theta/2-\phi})]$  and  $ q_{\phi}(i)=P(i|H_1)=\frac{1}{2}[1+(-1)^{i}\cos({\theta/2+\phi})]$, depending on which hypothesis is true.
 For simplicity we take equal priors $\eta_{0}=\eta_{1}=1/2$ and study the Bayesian mean number of copies under the same strong error bounds $\epsilon_{0}=\epsilon_{1}=\epsilon\ll 1$.
In the main text we show that the optimal test for a given choice of measurement angle is given by Wald's SPRT strategy, which according to (11) in MT leads to
\begin{equation}
\mean{N}=\eta_{0} \mean{N}_{0} +\eta_{1} \mean{N}_{1}\sim -\frac{1}{2}\log \epsilon  \left(\frac{1}{D(p_{\phi}\| q_{\phi})}
+\frac{1}{D(q_{\phi}\| p_{\phi})} \right) \,.
\label{Nlocal}
\end{equation}
\begin{figure}[htbp]
\begin{center}
\includegraphics[scale=.3]{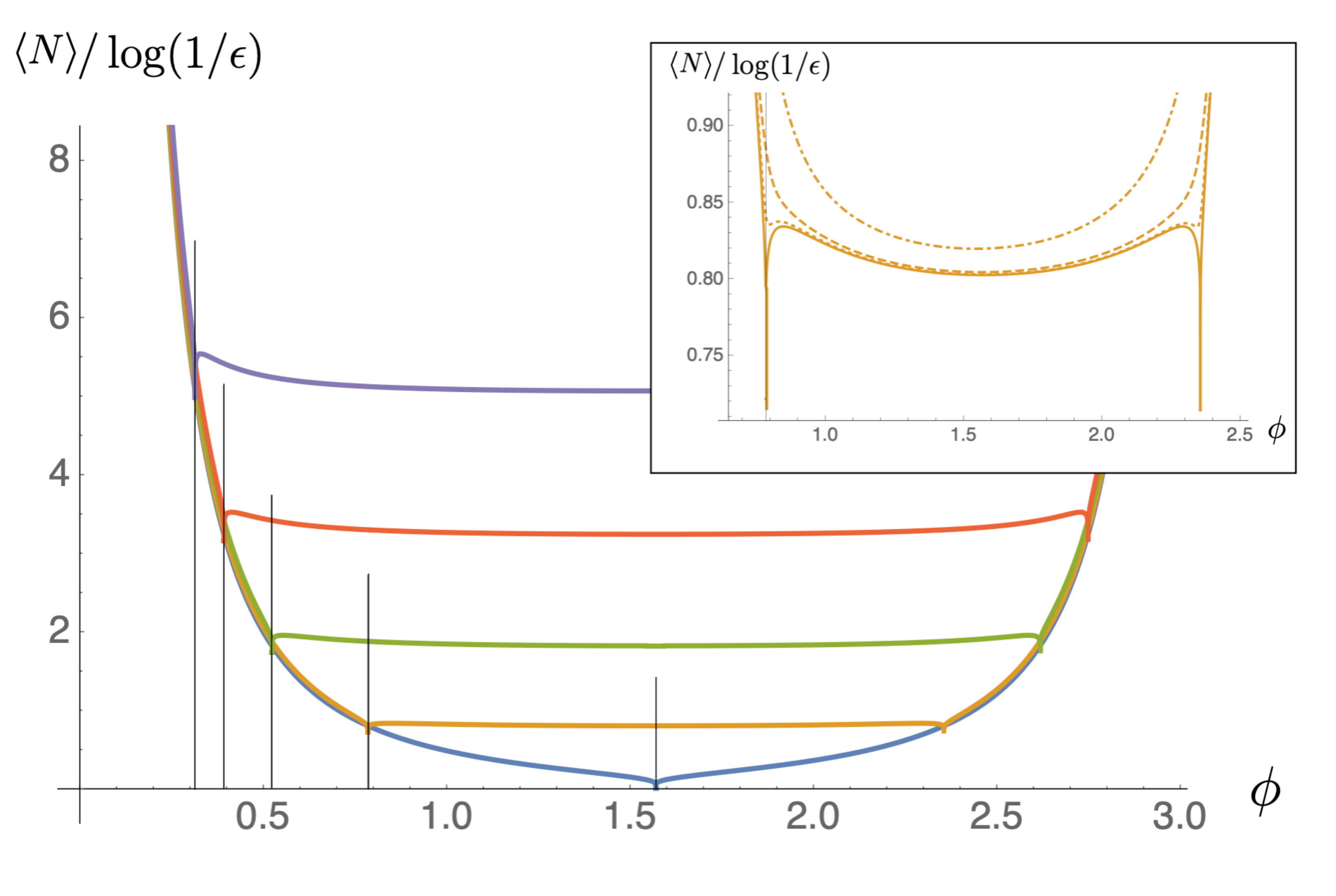}
\caption{Bayesian mean number of copies $\mean{N}$ as a function of the measurement angle $\phi$ for different pairs of pure states: from bottom to top 
$\theta=\{\pi,\frac{\pi}{2},\frac{\pi}{3},\frac{\pi}{4},\frac{\pi}{5}\}$.
 The vertical lines show the corresponding optimal measurement angles $\phi=\theta/2$. The inset shows in more detail the case with $\theta=\frac{\pi}{8}$, including the curves for noisy states with $1-r=\{10^{-2},10^{-3},10^{-4}\}$ (dashed lines, from top to bottom).}
\label{fig:OptimAngle}
\end{center}
\end{figure}

In Figure~\ref{fig:OptimAngle} we show the Bayesian mean number of copies required to have a guaranteed, asymptotically small bounded error $\epsilon$ for all outcomes of the experiment.
 For pure states ($r=1$), we observe that the optimal angle is a 
  singular point located at
   $\phi=\theta/2$, that corresponds to the fully biased measurement for which outcome 1 can only occur under hypothesis $H_{1}: p(1|H_{0})=0$ while  $p(1|H_{1})=\cos^{2}(\theta/2)>0$. Hence,  $H_{1}$  is detected with certainty after a small number of steps $\mean{N}_{1}\sim \cos^{-2}(\theta/2)$ (independent of the error bound $\epsilon$), and therefore the leading contribution to the expected number of copies when hypothesis $H_{0}$ is true is
\be
\label{eq:Npure}
\mean{N}_\text{local}\sim\frac{\log\epsilon}{2 \log(\cos^{2}\!\frac{\theta}{2})}\,.
\ee
Note that this is exactly half of  the number of copies that the most general collective deterministic strategy would require to attain this error bound, since $\epsilon=\frac{1}{2}(1-\sqrt{1-\cos^{2N}(\theta/2)})\sim \cos^{2N}(\theta/2)$. 
This error bound can be attained with local adaptive measurements for finite $N$ \cite{acin_multiple-copy_2005} and fixed local measurements for asymptotically large $N$.
 The result in \eqref{eq:Npure} is in agreement with that derived in \cite{Slussarenko2017MinASN}
for the fully biased strategy,  which we have shown to be optimal in the limit of small error bounds (among fixed local measurement strategies). In  Figure \ref{fig:OptimAngle} we also note that a small change around  
 the optimal value $\phi=\theta/2$ produces a very rapid increase of the effective number of copies while  the local minimum at $\phi=\pi/2$, which corresponds to the fully unbiased measurement, is much more shallow and hence more robust to a possible measurement misalignment.

We now proceed to study what happens in the presence of noise, when both states are mixed, in particular when $r_{0}=r_{1}=r$.
 As shown in the inset of Figure \ref{fig:OptimAngle}, the presence of noise makes the two states more indistinguishable and a higher number of samples are required to meet the error bound. It is also apparent that in presence of noise the fully unbiased measurement, $\phi=\pi/2$,   becomes optimal (except for extremely high values of the purity $1-r\sim10^{-5}$ for which fully biased performs slightly better). The unbiased measurement is straightforward to compute:
\be
\label{eq:Nunbiased}
\mean{N}_{\mathrm{local}}\sim\frac{\log\epsilon}{ r\sin\theta \log\left(\frac{1-r\sin\frac{\theta}{2}}{1+r\sin\frac{\theta}{2}}\right)}\,.
\ee

We can again compare  $\mean{N}$ reached by the local measurements \eqref{eq:Nunbiased}with the sample size, $N$, required by the optimal  deterministic protocol using a predetermined number of copies to achieve the same error $\epsilon$. When $N$ is large, i.e. 
$\epsilon$ is small, this can be obtained from the asymptotic error exponent in the quantum Chernoff 
bound~\cite{calsamiglia_quantum_2008,audenaert_discriminating_2007}. We find
\begin{equation}
    \label{eq:Ncher}
    N_{\rm Ch}\sim\frac{\log\epsilon}{ \log\left(1-(1-\sqrt{1-r^{2}})\sin^{2}\frac{\theta}{2}\right)}\,.
\end{equation}
In Figure 2 in MT we compare Eqs.~\eqref{eq:Nunbiased} and \eqref{eq:Ncher} and observe a reduction of the required number of copies of at least 50\% on average if we employ the sequential test instead of the deterministic one. The reduction 
goes up to 75\% if $\rho$ and $\sigma$ are very mixed.

\begin{figure}[htbp]
\subfloat[]{\includegraphics[width=.5\textwidth]{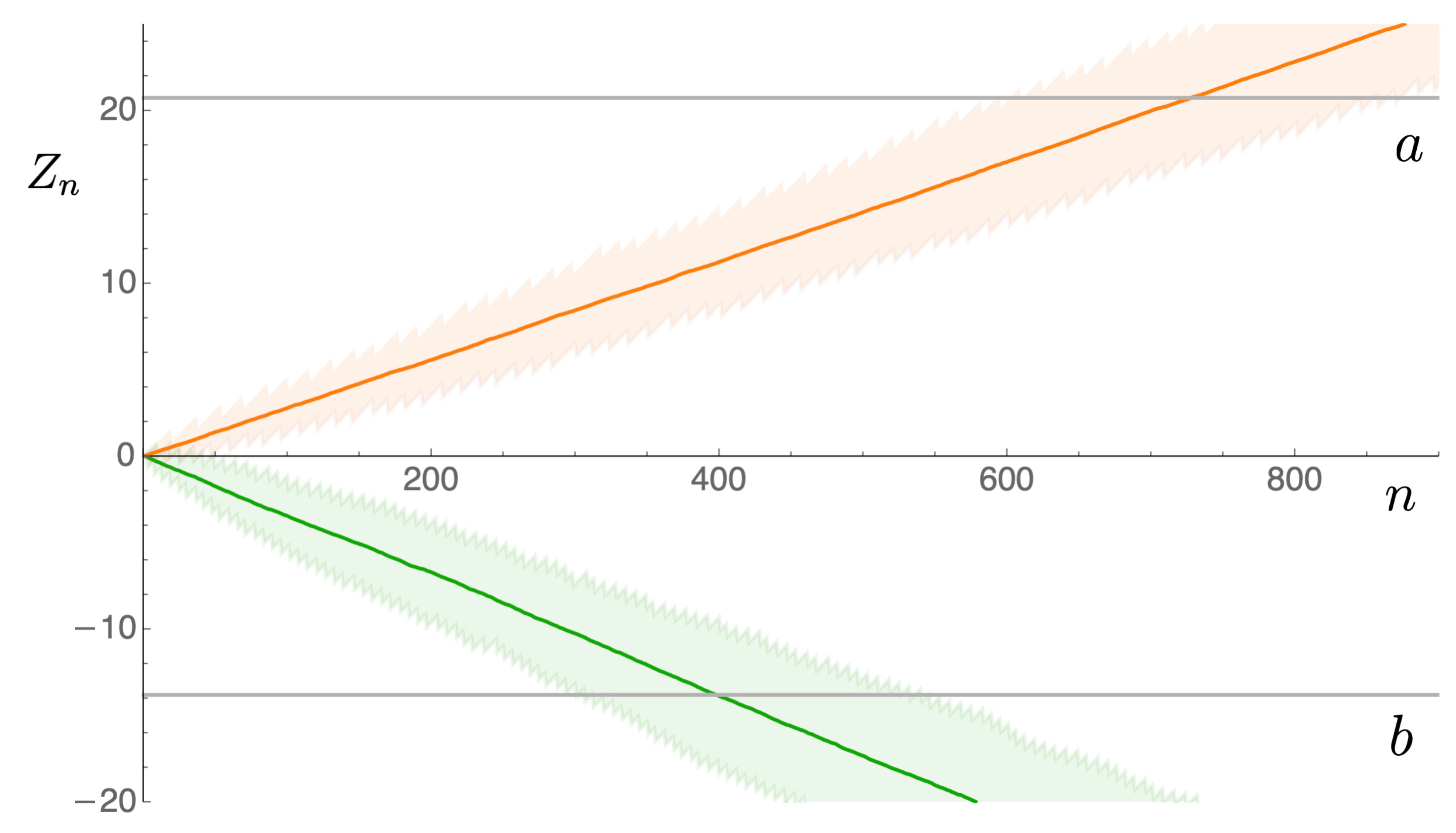}}
\subfloat[]{\includegraphics[width=.5\textwidth]{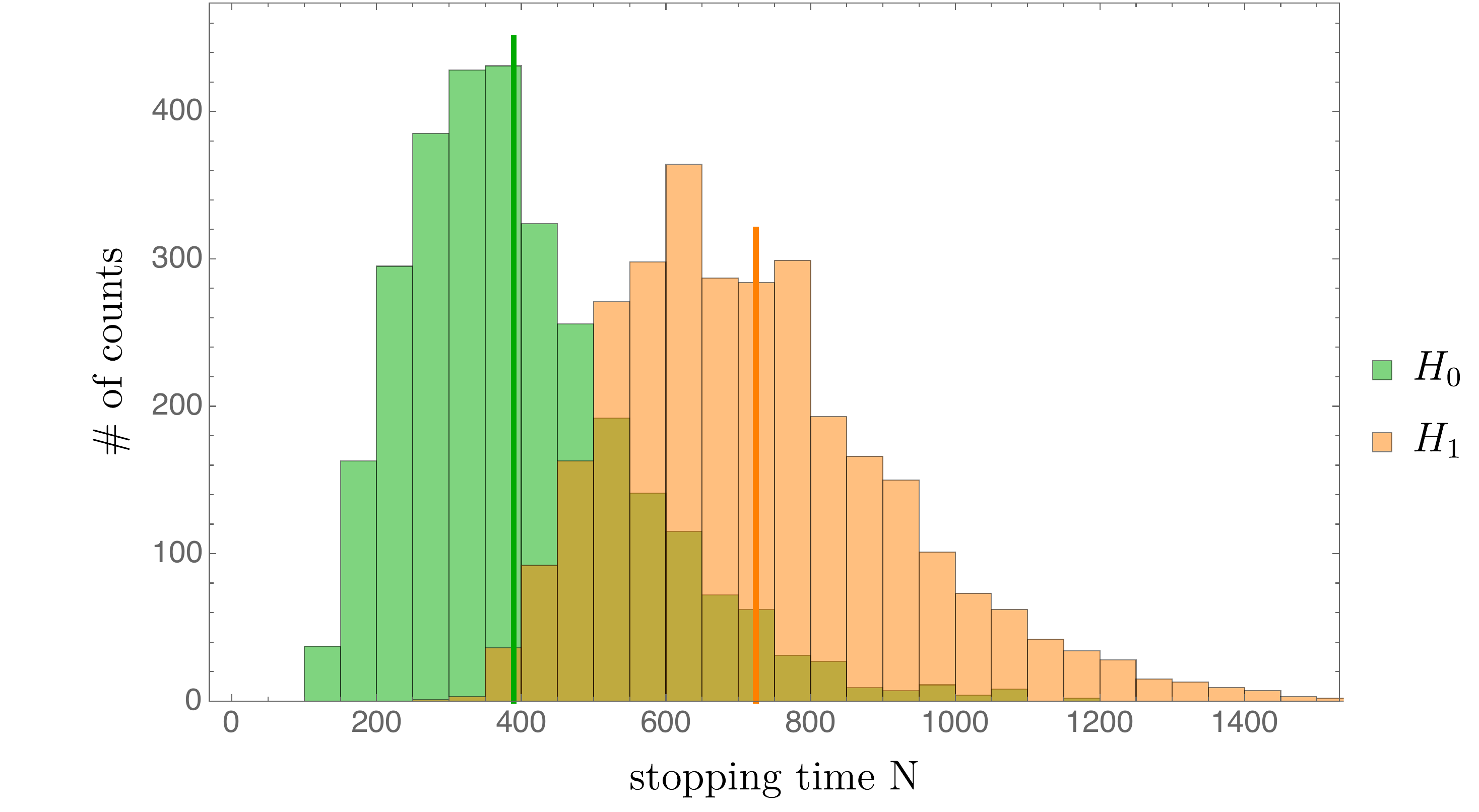}}\\
\subfloat[]{\includegraphics[width=.45\textwidth]{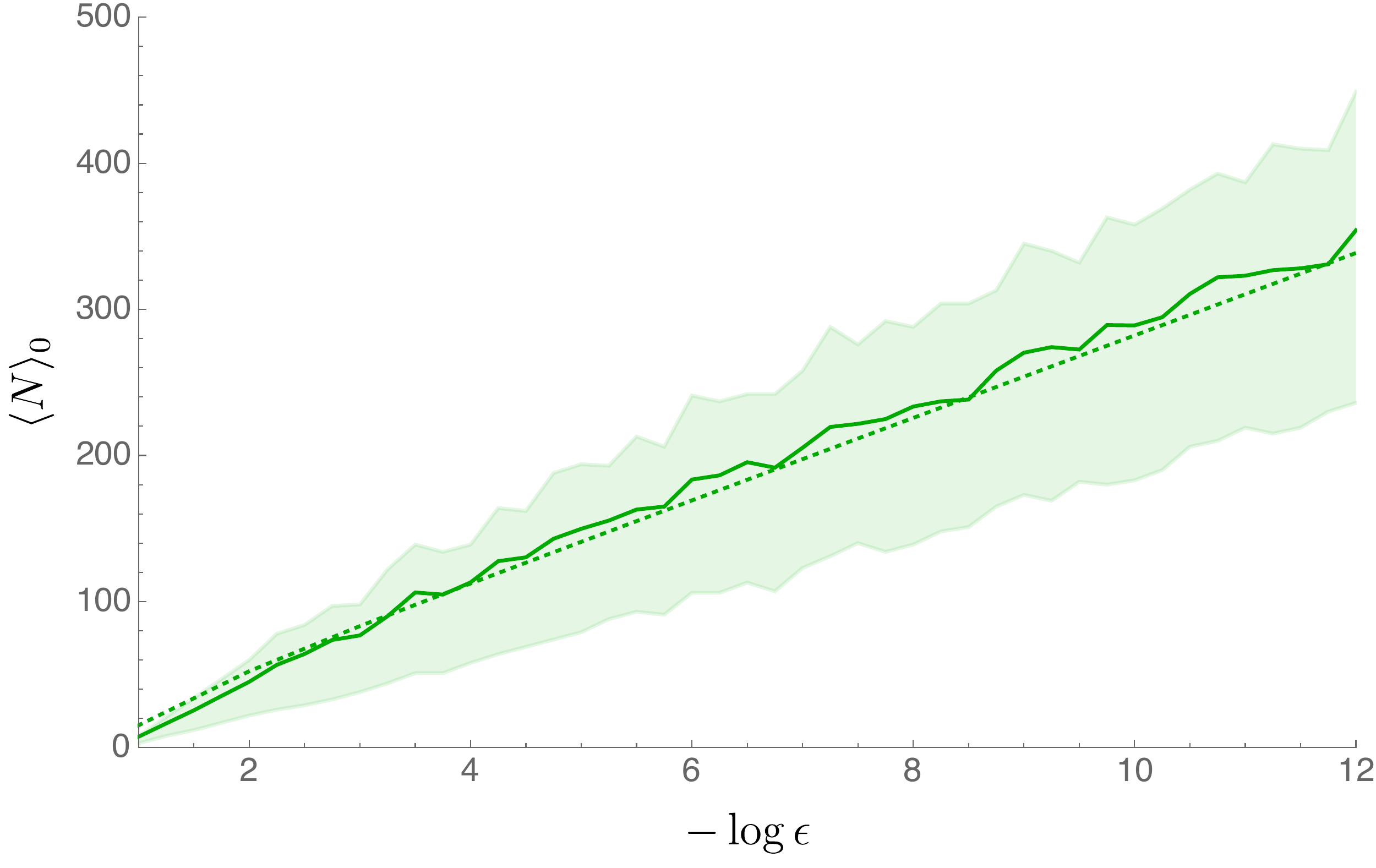}}
\caption{
(a) Behaviour of the likelihood ratio $Z_n$ using unbiased local measurements as a function of the step $n$, for 3000 realizations of the hypothesis test with $\theta={\pi\over 10}$, $r_0=0.7$, $r_1=0.9$, $\epsilon_0=10^{-6}$, $\epsilon_1=10^{-9}$, under hypotheses $H_0$ (green) and $H_1$ (orange). The solid lines show the mean trajectory, and the shaded areas correspond to the 2$^\text{nd}$ and 3$^\text{rd}$ quartiles (25\% above and below the mean). The decision boundaries $a$ and $b$ are also plotted.
(b) Histogram of stopping times under each hypothesis. The vertical lines mark the mean of each distribution.
(c) The mean number of copies $\mean{N}_0$ as a function of $\log\epsilon$, for the same state parameters as in (a) and equal error bounds $\epsilon_0=\epsilon_1=\epsilon$. The solid line represents the mean trajectory of 500 runs, the shaded area shows the 2$^\text{nd}$ and 3$^\text{rd}$ quartiles, and the dashed line is the analytic expression $\mean{Z}_0/D(p_\phi\|q_\phi)$, where recall that $\mean{Z}_0\sim -\log\epsilon$ for $\epsilon\ll 1$.
}
\label{fig:histo}
\end{figure}


For illustration purposes, in Figure~\ref{fig:histo} we show explicitly the results of several runs of a SPRT using unbiased local measurements. We observe how the mean trajectories that the cummulative log-likelihood ratio $Z_n$ follows point upwards or downwards depending on the underlying hypothesis. In this simulation, the state $\rho$, corresponding to $H_0$, is identified quicker than $\sigma$ (the decision boundary $b$ is closer than $a$), despite being more mixed. A histogram of stopping times under each hypothesis shows us that the distributions of $N$ are well-centered around their empirical mean, with right tails that are slightly longer; this is also apparent on the left figure from the cross-sections of the trajectories with the decision boundaries. Finally, we observe that the mean number of copies increases linearly with $\log\epsilon$ for $\epsilon\ll 1$, as predicted.

\subsection{Block-sampling and irrep projection}
Here we study the mean number of copies under both hypotheses using a block-sampling strategy where the same collective measurement is repeated on batches of $\ell$ copies.
In particular we will consider a collective measurement for which Hayashi  \cite{Hayashi_2001} 
showed that the (classical) relative entropy of the distributions that arise from it, attains the quantum relative entropy when the block length $\ell$ is large. 
Denoting by $M=\{M_{k}\}$ such collective POVM and by $\P^{\ell}_M, \Q^{\ell}_M$ the probability distributions of the outcomes, i.e., $\{\P^{\ell}_M(k)=\tr ( \rho^{\otimes \ell}M_{k})\}$ and $\{\Q^{\ell}_M(k)=\tr ( \sigma^{\otimes \ell}M_{k})\}$, in Ref.~\cite{Hayashi_2001}  it is shown that
\be
\frac{D\left(\P^{\ell}_M\| \Q^{\ell}_M\right)}{\ell} \leq D(\rho \| \sigma) \leq \frac{D\left(\P^{\ell}_M\| \Q^{\ell}_M\right)}{\ell} + (d-1)\frac{\log (\ell+1)}{\ell}\,,
\label{eq:MH}
\ee
where $d$ is the dimension of the underlying Hilbert space, from where
\be
\frac{D\left(\P^{\ell}_M\| \Q^{\ell}_M\right)}{\ell} \rightarrow D(\rho \| \sigma) \quad \mbox{ as } \ell \rightarrow \infty \,.
\label{eq:Dlim}
\ee
Quite remarkably, 
the measurement $M$ in Eq.~\eqref{eq:MH} depends solely on state $\sigma$.

As explained in the MT such a strategy allows one to attain the lower bound 
for one of the hypotheses, say $\mean{N}_0 \sim -\log\epsilon_0(1-A^{-1})/D(\rho\|\sigma)$ for the strong error bounds [cf. Eq.~(15) in MT], or
\be
\mean{N}_{0}\sim\frac{\log\beta(1-\alpha)}{D(\rho\| \sigma)}
\ee
for the asymmetric setting.

In what follows we compute the (sub-optimal) performance of this very same measurement under the other hypothesis, i.e., $\mean{N}_{1}$.

For qubit systems, the POVM that achieves the quantum relative entropy \cite{Hayashi_2001} corresponds to the simultaneous measurement of the total angular momentum $J^{2}$ (eigenspaces labeled by $j$) and its component along the axis $J_{z}$ (eigenspaces labeled by $m$), where $\hat{z}$ is picked to be the axis along which the state $\sigma$ points, i.e., $\sigma=1/2(\id+r_{1} \sigma_{z})$. The quantum number $j$ labels  the $SU(2)$ irreducible representations (irreps), and since it is invariant under the action of any rigid rotation $U^{\otimes n}$ it will only provide information about the spectra of $\rho$ or $\sigma$ ---which we denote $\lambda_{i}^{\pm}=\frac{1}{2}(1\pm r_{i})$ for $i=0,1$ respectively. 
The second measurement $J_{z}$ is clearly not $SU(2)$ invariant and provides information about relative angle between both hypotheses, and additional information on their spectrum.

Due to the permutational invariance of the $\ell$ copies it is possible to write the states in a block-diagonal form in terms of the $\{j,m\}$ quantum numbers (see e.g. \cite{gendra_beating_2012}):
\be
\sigma^{\otimes \ell}=\sum_{j=j_{\min}}^{\ell/2} q_{j} \sigma_{j}\otimes \frac{\id_{{j} }}{\nu_{j}} \quad \text{with}\quad
 \sigma_{j}=\sum_{m=-j}^{j} q(m|j)\ketbrad{j,m}  \mathrm{,} 
\ee
where  $\id_{{j}}$ are projectors over the subspaces of dimension 
$\nu_{j}={\ell \choose \ell/2-j}{2j+1\over \ell/2+j+1}$ that host the irreps of the permutation group (i.e. multiplicity space of spin $j$), $j_{\min}=0$ for $\ell$ even  ($j_{\min}=1/2$ for $\ell$  odd)  and 
\begin{eqnarray}
 \label{eq:qj}
q_{j}&=&\left(\frac{1-r_{1}^{2}}{4}\right)^{\ell/2} \nu_{j} Z_{j} \,,\\
 q(m|j)&=&\frac{1}{Z_{j}}R_{1}^{m}\quad \mathrm{with} \quad Z_{j}=\frac{R_{1}^{j}-R_{1}^{-j}}{R_{1}-1} 
 \quad \mbox{and} \quad R_{1}=\frac{1+r_{1}}{1-r_{1}}>1
 \label{eq:qmj}
\end{eqnarray}
are normalized probability distributions.

Under hypothesis $H_{0}$ the state has exactly the same  structure  except for a global rotation around 
the $\hat z$ axis by an angle $\theta$,
\begin{equation}
\rho^{\otimes \ell}=\sum_{j=j_{\min}}^{\ell/2} p_{j} \rho_{j}\otimes \frac{\id_{{j}} }{\nu_{j}} \quad \text{with}\quad
 \rho_{j}=\sum_{m=-j}^{j} p(m|j)\;U_{\theta}\ketbrad{j,m} U^{\dagger}_{\theta} \mathrm{,} \label{eq:rhoiid}
\end{equation}
where $p_{j}$ and $p(m|j)$ take the form of  \eqref{eq:qj} and \eqref{eq:qmj} replacing $r_{1}$ and $R_{1}$ by $r_{0}$ and $R_{0}$.

The outcomes of the  $J^{2}$ and $J_{z}$ measurements lead to probability distributions
\begin{align}
p(j,m)&=p_j\,\tilde{p}(m|j) \quad \mathrm{ with }\quad \tilde p(m|j)=\sum_{m'=-j}^{j}p(m'|j) |\!\bra{j,m} U_{\theta}\ket{j,m'}\!|^{2} \,,\\
q(j,m)&=q_j\,q(m|j) \,,
\end{align}
whose relative entropy can be written as
\begin{align}
D(q^{\ell}\| p^{\ell})=&\sum_{j=j_{\rm min}}^{\ell/2}q_j \sum_{m=-j}^j q(m|j) \log\frac{q_j\,q(m|j)}{p(j,m)}\sim \log\frac{q_{j^{*}}\,q(j^{*}|j^{*})}{p(j^{*},j^{*})}
 \label{eq:dpq}
\end{align}
where  we have used the fact that for $\ell\gg 1$, $q_j$  
is strongly peaked at $j^{*}=r_{1}\ell /2$ 
and  $q(m|j^{*})$  
decays exponentially, and hence it is peaked at $m=j^{*}$. 
In addition we note that
\begin{align}
p(j,m=j)&=p_j\sum_{m'=-j}^j p(m'|j) |\!\bra{j,j}U_\theta\ket{j,m'}\!|^{2}\nonumber\\
&=
\bra{j,j}\left( p_j \sum_{m'=-j}^j p(m'|j) U_\theta \ketbra{j,m'}{j,m'} U_\theta^\dagger \right)\ket{j,j}
=\left(\frac{1-r_0^2}{4}\right)^{l/2-j}\nu_j \bra{j,j} \rho^{\otimes 2j }\ket{j,j}\nonumber\\
&=\left(\frac{1-r_0^2}{4}\right)^{l/2-j}\nu_j\bra{\uparrow}\rho\ket{\uparrow}^{2 j}=
\left(\frac{1-r_0^2}{4}\right)^{l/2-j}\nu_j\left(\frac{1+r_0\cos\theta}{2}\right)^{2j} \,,
\end{align}
where we used the general decomposition of \eqref{eq:rhoiid}, and $\ket{\uparrow}$ is short-hand notation for $\ket{j=1/2,m=1/2}$.
Inserting this expression in  \eqref{eq:dpq} and using the definitions in \eqref{eq:qj} and
 \eqref{eq:qmj} we finally arrive at 

\begin{align}
D_{M_{\sigma}}(\sigma\| \rho)&:= \lim_{\ell \to\infty} \frac{1}{\ell}D(q^{\ell}\| p^{\ell})
= \frac{1}{2}(1-r_{1}) \log \frac{1-r_{1}^{2}}{1-r_{0}^{2}}+r_{1} \log \frac{1+r_{1}}{1+r_{0} \cos\theta}\nonumber\\
&=D(\lambda_{1}\|\lambda_{0}) +r_{1} \log\frac{1 + r_{0} }{1 + r_{0} \cos\theta} \,,
\label{eq:DsrMs}
\end{align}
where the symbol $M_\sigma$ recalls that we have chosen $M$ as a measurement over the eigenbasis of $\sigma$, which maximizes the relative entropy $D(q^l\|p^l)$, and
\be
D(\lambda_{1}\|\lambda_{0})=\frac{1}{2}(1+r_{1}) \log \frac{1+r_{1}}{1+r_{0}}+
\frac{1}{2}(1-r_{1}) \log \frac{1-r_{1}}{1-r_{0}}
\ee
 is the relative entropy between the spectra of $\sigma$ and $\rho$. Hence, the second term in \eqref{eq:DsrMs}
can be associated to the distinguishability caused by the different orientation (non-commutativity) of the states.

On the other hand, from \eqref{eq:Dlim} it follows
\begin{align}
D_{M_{\sigma}}(\rho\| \sigma)&:= \lim_{\ell \to\infty} \frac{1}{\ell}D(p^{\ell}\| q^{\ell})=D(\rho\| \sigma) \nonumber\\
&=\frac{1}{2}(1+r_{0}) \log \frac{1+r_{0}}{1+r_{1}}+
\frac{1}{2}(1-r_{0}) \log \frac{1-r_{0}}{1-r_{1}}+
r_{0} \log \frac{1+r_{0}}{1-r_{1}} \sin^{2}\frac{\theta}{2} \nonumber\\
&=D(\lambda_{0}\|\lambda_{1})+r_{0} \log \frac{1+r_{0}}{1-r_{1}} \sin^{2}\frac{\theta}{2} \,.
\label{eq:Drs2}
\end{align}

From the above results we conclude that applying the measurement that reaches
the ultimate bound for one hypothesis
\be
\mean{N}_{0}=-\frac{\log\beta}{D(\rho\| \sigma)}
\label{eq:No0}
\ee
will result in a sub-optimal value 
\be
\mean{N}_{1}\sim-\frac{\log\alpha}{D_{M_{\sigma}}(\sigma\| \rho)}
\ee
for the other hypothesis, with $D_{M_{\sigma}}(\sigma\| \rho)$ given in \eqref{eq:DsrMs}.

We observe that, as expected, when the states commute we can reach the
 ultimate bound for both $\mean{N}_{0}, \mean{N}_{1}$. 
 We also note that, when $\rho$ is pure, one can also preserve  asymptotic optimality 
 for both means, since when $\lambda_{0}^{\pm}\in\{1,0\}$, $D_{M_{\sigma}}(\sigma\| \rho)$ diverges and the leading contribution in $\mean{N}_{1}$ vanishes, while $\mean{N_{0}}$ reaches the optimal  value.
 These results hold for the block-sampling strategy that uses blocks of large length $l\gg 1$, so one needs to 
 find other ways to compute the finite $O(1)$ contribution to $\mean{N}_{1}$, as we shall show next.
 
We have already shown [see (18) in MT] that when \emph{both} states are pure we can  detect both hypotheses with a finite mean number of copies. If only one of the states is pure, say $\rho$, it is easy to notice that the $J^{2}$ measurement alone guarantees a constant value for $\mean{N}_{1}$: $\rho^{\otimes n}$ lies in the fully symmetric space
(with $j=n/2$) and therefore any measurement outcome $j<n/2$ will unambiguously identify $\sigma$.
The above  block-sampling might have an important overhead when $\ell$ is large.
A way of reducing this overhead can be devised by leveraging the fact that the measurement of $J^{2}$ on $n$ copies commutes with the measurement of $J^{2}$ on $n+1$ for all $n\geq 1$: 
starting at $n=2$, we measure $J^{2}$ sequentially on all available copies until we get an outcome $j<n/2$, at which moment we stop and accept $H_{1}$.
Note that each step of this sequence uses the already measured copies, increasing the number of jointly-measured systems by one.
Since the probability of not detecting $\sigma$ at step $n$ 
(continue measuring) is given by
$T^{n}_{1}=P(j=n/2)=\left(\frac{1+r_1\cos\theta}{2}\right)^{n}=:q_{+}^{n}$,
we can write
\be
\mean{N}_{1}=\sum_{n=0}^{\infty}T^{n}_{1}=\sum_{n=0}^{\infty}q_{+}^{n}=
\frac{1}{1-q_{+}}=\frac{2}{1-r_1\cos\theta} \,.
\label{eq:aa}
\ee

Note, however, that measuring $J^{2}$ sequentially is on its own not enough to reach the optimal mean number of copies also under hypothesis $H_0$, $\mean{N}_{0}$. 
For this reason, after every batch of $\ell$ copies, $\ell\gg \mean{N}_{1}$ ($\ell=o(\log1/\epsilon)$),
we interrupt the sequence of $J^{2}$ measurements with a measurement of $J_{z}$ on the last batch of $\ell$ copies (and then continue again with the $J^{2}$ sequential measurement). The measurement statistics  obtained  by this procedure mimicks the block-sampling method described above and hence we are guaranteed to converge to \eqref{eq:No0}.

Alternatively, in order to attain  \eqref{eq:aa} one can directly measure the system in the basis that diagonalizes $\rho=\ketbrad{\lambda_{0}^{+}}$, so that an  $\ket{\lambda_{0}^{-}}$ outcome  unambiguously detects $\sigma$ with probability $q_{-}=1-q_{+}$. It is immediate to check that the sequential application of this measurement also leads to  \eqref{eq:aa}. Again after having measured a sufficiently large number of copies $\ell=o(\log1/\epsilon)$ one can adopt the block-sampling strategy in order to achieve the bound  \eqref{eq:No0}.

\subsection{Ultimate limit for $N_{\mathrm{wc}}$. Attainability regions}
In this section we study the achievability of the lower bound on the worst-case mean number of copies, 
\begin{equation}
    N_{\mathrm{wc}}:=\max\left\{\mean{N}_{0},\mean{N}_{1}\right\}\geq \max\left\{-\frac{\log\epsilon}{D(\rho\| \sigma)},-\frac{\log\epsilon}{D(\sigma\| \rho)}\right\},
    \label{eq:wcbound}
\end{equation}
where for simplicity we assume that both error bounds are equal, i.e., $\epsilon_{0}=\epsilon_{1}=\epsilon\ll 1$.

In the previous section we have shown that the block-sampling with a given POVM $M_{\sigma}$  can reach  the optimal value under $H_{0}$, but it does so at the expense of attaining a sub-optimal value  under $H_{1}$. 
Making use of the results of \eqref{eq:DsrMs} and \eqref{eq:Drs2} one can show  that there are 
pairs of states $\{\rho, \sigma\}\in \mathcal{R}$ for which either
\begin{equation}
\label{eq:condsat}
D(\rho\|\sigma)\leq D_{M_{\sigma}}(\sigma\| \rho) \leq D(\sigma\|\rho) \quad\text{or}\quad
D(\sigma\| \rho)\leq D_{M_{\rho}}(\rho\| \sigma) \leq D(\rho\| \sigma)
 \,.
 \end{equation}
When this happens we can assert that the bound in \eqref{eq:wcbound} is attainable, since the worst-case value is attained, i.e.,
\begin{equation}
    N_{\mathrm{wc}}\sim \max\left\{-\frac{\log\epsilon}{D(\rho\| \sigma)},-\frac{\log\epsilon}{D(\sigma\| \rho)}\right\}\,,
    \label{eq:wcbound2}
\end{equation}

Figure \ref{fig:tight_regions} shows  some representative regions where \eqref{eq:condsat} is fulfilled, and \eqref{eq:wcbound2} holds. We observe that for small relative angles $\theta$ almost all states  attain the ultimate bound, except for a region around the pairs of equal purity.  It is easy to check that
for states with $r_{0}=r_{1}$, $D(\rho\|\sigma)=D(\sigma\|\rho)$, and therefore \eqref{eq:condsat} cannot be satisfied, independently of the relative angle $\theta$.
 When $\theta=\pi/2$, i.e., when the pair of states exhibits more non-classicality, only pairs comprised by a highly pure and a highly mixed state can attain the bound.
\begin{figure}[ht]
    \includegraphics[scale=0.45]{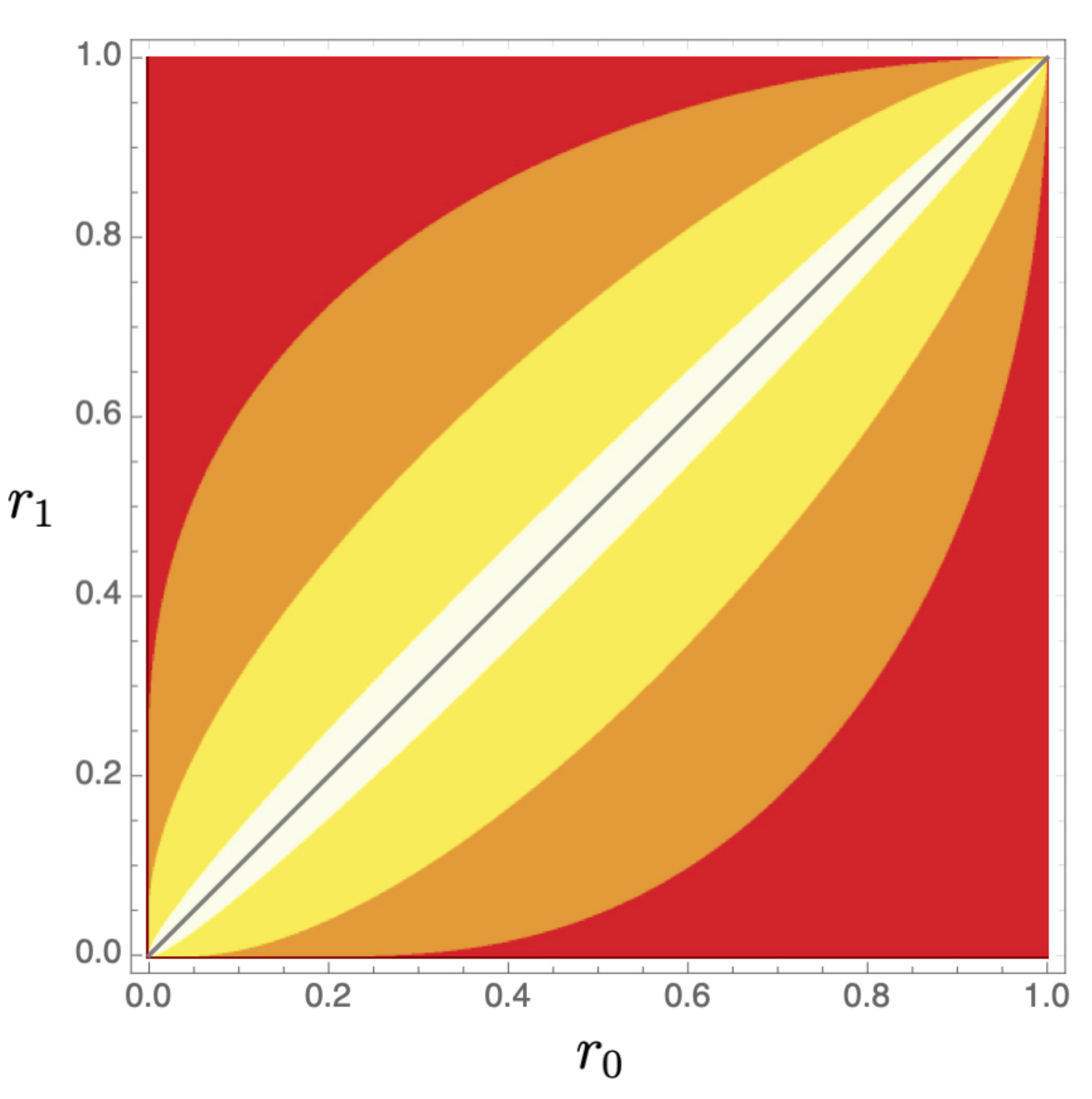}
    \caption{Regions $\mathcal{R}$ of purities $(r_{0},r_{1})$ for different relative angles $\theta$, where the ultimate bound for $N_{\mathrm{wc}}$ \eqref{eq:wcbound2} is achievable. The red region corresponds to $\theta=\pi/2$, orange and red to $\theta=\pi/10$,
 yellow, orange and red to $\theta=\pi/100$, and all admissible values of $(r_{0},r_{1})$ are achievable for commuting states ($\theta=0$).}.
\label{fig:tight_regions}
\end{figure}
\section{The overhead for arbitrary dimensions}
\label{sec:finitedim}

In this section we would like to explore how conditions \eqref{eq:condsat} look when states $\rho$ and $\sigma$ have arbitrary dimension $d>2$. In this case, exactly quantifying $D_{M_\rho}(\rho\|\sigma)$ [recall that we denote by $M_\rho$ the block-sampling measurement on $\ell$ copies that attains $D(\sigma\|\rho)$ when $\ell\to \infty$] is more involved.
Here instead we provide a general lower bound for 
the deviation of $D_{M_\rho}(\rho\|\sigma)$  from its maximum value $D(\rho\|\sigma)$. We follow closely Ref.~\cite{Hayashi_2001}. 
%

 First, consider the following operation on a state $\rho$ for a given a projective measurement $E=\{E_j\}$ (i.e., $E_j^2=E_j$ and $E_jE_k=\delta_{jk} E_j$),
 \begin{equation}
        \varepsilon_E(\rho) := \sum_j E_j\rho E_j \,.
    \end{equation}
When $E$ commutes with states $\rho$ and $\sigma$ we have 
 \be
 \varepsilon_E(\rho)=\rho \,,\quad \varepsilon_E(\sigma)=\sigma\,.
 \ee
 %
 %
 
Then, consider a projective measurement  $F(\rho)=\{F_k\}$ that consists of rank-one projectors in the eigenbasis of $\rho$, i.e., a measurement of the spectrum of $\rho$.  Note that we have $\varepsilon_F(\rho)=\rho$,
but  $\varepsilon_F( \sigma) \neq \sigma$ for a generic state $\sigma$ that does not commute with $\rho$. Note also that $E$ (which commutes with $\rho$ and
 $\sigma$) is a coarse-grained measurement of $F(\rho)$; we can then say that $F(\rho)$ is stronger than $E$. In~\cite{Hayashi_2001} this fact is denoted by $F(\rho)\geq E$. We then have the following lemma:

\begin{lemma}\label{lemma2}
    Let $\rho$ and $\sigma$ be states and let 
    $F(\rho)\geq E$. 
    The quantum relative entropy between $\rho$ and $\sigma$ can be expressed as 
    \begin{equation}
        D(\rho\| \sigma) = D(\varepsilon_F(\rho)\| \varepsilon_F(\sigma)) 
        +\Tr\rho(\log\varepsilon_F(\sigma)-\log\sigma)\,.
    \end{equation}
\end{lemma}
\begin{proof}
 Recalling that $\varepsilon_F(\rho)=\rho$,  we have $\Tr \varepsilon_F(\rho)\log\varepsilon_F(\sigma) = \Tr\rho\log\varepsilon_F(\sigma)$, thus
    \begin{align}
        D(\varepsilon_F(\rho)\| \varepsilon_F(\sigma))-D(\rho\| \sigma) 
= \Tr\varepsilon_F(\rho)(\log\varepsilon_F(\rho)-\log\varepsilon_F(\sigma))
- \Tr\rho(\log\rho-\log\sigma)
=\Tr\rho(\log\sigma-\log\varepsilon_F(\sigma))\,.
    \end{align}
Hence, it follows that
    \begin{equation}
         D(\varepsilon_F(\rho)\| \varepsilon_F(\sigma)) =D(\rho\| \sigma) -
        \Tr\rho(\log\varepsilon_F(\sigma)-\log\sigma)\,.
    \end{equation}
\end{proof}

%
We also need the following lemma:
\begin{lemma}\label{lemma3}
    For a given projective measurement $E$ such that $E\leq F$, if $E$  commutes with $\sigma$ and $\rho$ we have that
    \begin{equation}\label{eq:lemma3}
        \Tr\rho(\log\varepsilon_F(\sigma)-\log\sigma)\leq \sup_i\{\Tr\rho_i(\log\varepsilon_F(\sigma_i)-\log\sigma_i)\} \,,
    \end{equation}
    where we define $\rho_i:=\frac{1}{a_i} E_i\rho E_i$, $\sigma_i:=\frac{1}{b_i} E_i\sigma E_i$, $a_i:=\Tr E_i\rho E_i$, and $b_i:=\Tr E_i\sigma E_i$.
\end{lemma}

\begin{proof}

Starting from the left side of inequality~\eqref{eq:lemma3}, the following steps hold:
    \begin{align}
        \Tr\rho(\log\varepsilon_F(\sigma)-\log\sigma) &= \Tr[\sum_iE_i\rho(\log\varepsilon_F(\sigma)-\log\sigma)]
       = \Tr[\sum_iE_i\rho E_i(E_i\log\varepsilon_F(\sigma)E_i-E_i\log\sigma E_i)]\\
        &= \Tr[\sum_i a_i\rho_i(\log\varepsilon_F(\sigma_i)-\log\sigma_i)]
        \leq \sup_i\{\Tr\rho_i(\log\varepsilon_F(\sigma_i)-\log\sigma_i)\}
        =: \omega(\sigma).
    \end{align}
\end{proof}

We are now ready to derive a bound on $D_{M_\rho}(\rho\|\sigma)$ with the following theorem:
\begin{theorem}
    Let us define the projective measurement $M_\rho=F(\rho^{\otimes \ell})\times E^\ell$ acting on $\ell$ copies, where $E^\ell$, applied first, is a measurement that projects onto the irreps of $SU(d)^{\otimes \ell}$.
    Then, $F(\rho^{\otimes \ell})$ is a spectral measurement of $\rho^{\otimes \ell}$, i.e., a projective measurement on the basis that diagonalizes $\rho$.
For this measurment, we have
\begin{equation}
    D(\rho\| \sigma)-\frac{\omega(\sigma)}{\ell}\leq\frac{1}{\ell}
    D_{M_\rho}(\rho^{\otimes \ell}\| \sigma^{\otimes \ell})\leq D(\rho\| \sigma) \,.
\end{equation}
\end{theorem}

\begin{proof}
    We simply use Lemma~\ref{lemma2} and Lemma~\ref{lemma3}  with the change
    $\rho\rightarrow\rho^{\otimes \ell}$ and $\sigma\rightarrow\sigma^{\otimes \ell}$, and we recall the property of the quantum relative entropy $D(\rho^{\otimes \ell}\| \sigma^{\otimes \ell})= \ell D(\rho\| \sigma)$. The result follows from applying  Lemma~\ref{lemma2}  to all the terms in  Lemma~\ref{lemma3}. 
\end{proof}

A very generous bound  can be obtained by dropping the 
negative term $\log\varepsilon_F(\sigma_i)$ in $\omega(\sigma)$~\cite{audenaert_eisert_2005}:
\begin{align}
\omega(\sigma)\leq \sup_i-\Tr[\rho_i (\log\sigma_i)] \leq \max_i-\log [\lambda_{\min}(\sigma_i)]\,.
\end{align}

\section{Zero-error protocol for pure states}
\label{sec:unambiguous}

As we have seen in the MT, as the error $\epsilon$ goes to $0$, the average number
of copies goes to infinity. However, for
pure states $\{ \rho=\ketbrad{\psi_0}, \sigma=\ketbrad{\psi_1} \}$  there are sequential strategies with local measurements that give a strictly zero error with a finite 
average number of samples.  Here we detail the protocol already mentioned in the MT and prove its optimality for equal priors for the Bayesian mean and worst-case number of copies.

To this end, consider a sequence of fixed unambiguous measurements on each copy with  inconclusive probabilities $c_\nu$ if the given state is $\ket{\psi_\nu}$, $\nu=0,1$. We notice that these probabilities satisfy the  'uncertainty' relation $c_0 c_1\geq s^2$, where $s=|\braket{\psi_0}{\psi_1}|$ \cite{sentis2018online}.
The protocol stops only if one of the  states is identified with no error. Hence, at each step $n$ there are only two possibilities: continue, with conditional probability (after having arrived at step $n$) $c_\nu$, 
or stop, with conditional probability $1-c_\nu$. 
The probability of exactly stopping  at step $n$ is $P^n_\nu=c_\nu^{n-1}(1-c_\nu)$. 
Then, the  average number of copies required  to get a  zero-error outcome  is 
\begin{align}
\label{Mp}
\mean{N}_\nu=&  \sum_{n=1}^{\infty} n c_\nu^{n-1}(1-c_\nu)= \sum_{n=0}^{\infty} c_\nu^{n}=\frac{1}{1-c_\nu} \,.
\end{align}
Notice that both means are finite if one performs an  unambiguous measurement with $c_0<1$ and $c_1<1$, which is allowed by the relation $c_0 c_1\geq s^2$.

In the case of equal priors, we now show that the symmetric choice $c_0=c_1=s$ 
gives the optimal Bayesian mean $\mean{N}=(\mean{N}_0+\mean{N}_1)/2$. 
%
%
We observe that the inconclusive probability attained by the optimal global measurement on $n$ copies of $\ket{\psi_\nu}$, $T_\nu^{n}$, 
cannot be beaten by any local strategy, hence, using \eqref{Mp} we have
\begin{align}
\label{Mp-opt}
\mean{N}=\frac{1}{1-s}=\sum_{n=0}^{\infty} s^n\geq \sum_{n=0}^{\infty}  \frac{1}{2}\left(T_0^{n}+T_1^{n}\right) =:\mean{N^*}\,.
\end{align}
Analogously to the derivation of (15) in MT, the r.h.s. of \eqref{Mp-opt} corresponds to a relaxation of the original problem in which we have independently optimized each term in the sum, considering the action of optimal $n$-copy unambiguous measurements for each $n$, hence $\mean{N^*}$ is a lower bound to the most general protocol. 
Since $n$-copy pure states are simply pure states of larger dimension, we have $|\braket{\psi_0^{\otimes n}}{\psi_1^{\otimes n}}|=s^n$, and the equivalent relation 
$T_0^{n}T_1^{n}\geq s^{2n}$ 
holds for global strategies. 
%
Then, the symmetric choice $T_0^{n}=T_1^{n}=s^n$  
minimizes each summand in~\eqref{Mp-opt}, and we obtain $\mean{N^*} = \sum_{n=0}^\infty s^n = (1-s)^{-1} = \mean{N}$. 

Finally, we also note that the symmetric choice also optimizes the figure of merit given by the  worst-case number of copies $N_\mathrm{wc} =\max\{\mean{N}_0,\mean{N}_1  \}$.  Because of the relation 
$T_0^n T_1^n\geq s^{2n}$, if $T_0^n>s^n$ then $T_1^n<s^n$ and $N_\mathrm{wc}=\mean{N}_0>1/(1-s)$. 
The same argument applies if $T_1^n>s^n$,  
hence it follows  that  the optimal measurement has $T_0^n=T_1^n=s^n$ 
and  $N_\mathrm{wc}=1/(1-s)$.
%
%


\end{widetext}

\end{document}